%% file: ms.tex
\documentclass[a4paper,preprint]{elsarticle}

\usepackage{lineno}
\PassOptionsToPackage{normalem}{ulem}
\usepackage{ulem}

\usepackage{algorithm}
\usepackage{algpseudocode}
\usepackage{graphicx}
\usepackage{epsfig}
\usepackage{caption}
\usepackage{subcaption}
\usepackage{amsmath}
\usepackage{xparse}
\usepackage{amssymb}
\usepackage{tcolorbox}
\usepackage{amsthm}
\usepackage{graphicx}

\usepackage{cleveref}
\usepackage{enumitem}

\newtheorem{definition}{Definition}

\newtheorem{lemma}{Lemma}

\newtheorem{corollary}{Corollary}

\newcommand{\Cost}{\mathrm{Cost}}
\newcommand{\Trunc}{\mathrm{Trunc}}
\newcommand{\sig}{\mathrm{sig}}
\newcommand{\depth}{\mathrm{depth}}
\newcommand{\Expand}{\mathrm{Expand}}
\newcommand{\OPT}{\mathrm{OPT}}
\newcommand{\Pred}{\mathrm{Pred}}
\newcommand{\CALI}{{\cal I}}
\newcommand{\CALS}{{\cal S}}

\newenvironment{Note}
{\em \small
Note: 
}
\DeclareMathOperator*{\argmin}{\arg\!\min}

\title{Speeding up the AIFV-$2$ dynamic programs by two orders of magnitude using Range Minimum Queries}

\author[HKUST1]{Mordecai Golin \fnref{fn1}}
\address[HKUST1]{Hong Kong  UST.  golin@cse.ust.hk}

\author[HKUST2]{Elfarouk Harb}
\address[HKUST2]{Hong Kong  UST.  eyfmharb@connect.ust.hk}
\fntext[fn1]{Work partially supported by Hong Kong RGC CERG  Grant	 16213318}
\date{}

\begin{document}

\begin{abstract}
AIFV-$2$ codes are a new method for constructing lossless codes for memoryless sources that provide better worst-case  redundancy than Huffman codes.  
They do this by using two code trees instead  of one and also allowing some bounded delay in the decoding process.  
Known algorithms for constructing AIFV-code are iterative;  at each step they  replace the current code tree pair with  a ``better'' one.  The current state of the art for performing this replacement is a pair of  Dynamic Programming (DP) algorithms that use $O(n^5)$ time   to fill in two tables, each of size $O(n^3)$ (where $n$  is the  number of different characters in the  source).

This paper describes how to reduce the time for filling in the DP tables by two orders of magnitude,  down to $O(n^3)$. It does this by introducing a grouping technique that permits separating the $\Theta(n^3)$-space tables into $\Theta(n)$ groups, each of size $O(n^2)$, and then using Two-Dimensional Range-Minimum Queries (RMQs) to fill in that group's  table entries in $O(n^2)$ time.  This RMQ speedup technique seems to be new and might be of independent interest.

\end{abstract}

%keywords
% Tree Partitioning
% Sink Evacuation
% MInmax-Regret
% $k$-Center
% Confluent Flows

\begin{keyword}
AIFV Codes \sep
Dynamic Programming Speedups   \sep
Range Minimum Queries
% Minmax-Regret   \sep
%$k$-Center   \sep
%Confluent Flows
\end{keyword}

\maketitle
% Introduction part including overview of the paper, past work and motivation
\section{Introduction}
\label{sec:Intro}
\input{Introduction.tex}

\section{Range Minimum Queries}
\label{sec:RMQ}
\input{RMQ.tex}

\section{The Dynamic Program and its speedup}
\label{sec:The DP}
\input{DP_and_Speedup_v2.4.tex}

\section{A Quick Introdution to AIFV-$2$ codes}
\label{sec:The Codes}
\input{AIFV_Intro.tex}

\section{ Deriving the DP}
\label{sec:Derivation}
\input{Derivation_v2.tex}

\medskip

\par\noindent{\bf References:}

% Bibliography. The file is main.bib
%\bibliographystyle{splncs03} 
\bibliographystyle{plain} 
\bibliography{biblio_v4}

%\bibliography{Evac,Evacuation_Problems}

\end{document}

%% file: Introduction.tex
%\section{Introduction}
%\label{sec:Intro}

Almost Instantaneous Fixed to Variable-$2$ (AIVF-$2$) codes were introduced recently in a series of papers \cite{introduceMAIFV,dp,original,yamamoto2015almost}

Similar to Huffman Codes, these provide lossless encoding for a fixed probabilistic memoryless  source.  They differ from Huffman codes in that they use a {\em pair} of  coding trees instead of just   one tree, sometimes coding using the first and sometimes using the second.   They also no longer provide {\em instantaneous} decoding.  Instead, decoding might require a bounded delay.  That is,  it might be necessary to read up to $2$ extra characters after a codeword ends before certifying the the completion (and decoding) of the codeword.
 The advantage of AIFV-$2$ codes over Huffman codes is that they guarantee  redundancy of at most $1/2$   instead of the redundancy of  $1$ guaranteed by Huffman encoding \cite{{introduceMAIFV}}.

The procedure for constructing optimal (min-redundancy) 
AIFV-$2$ codes is much more complicated than that of finding Huffman codes.
It is an iterative one that, at each step,  replaces the current pair of coding trees by  a new, better, pair.  The original  paper   \cite{yamamoto2015almost} only proved  that its iterative algorithm terminated.  This was improved to polynomial time  steps  by 
\cite{golin2019polynomial},  which used  only $O(b)$ iterations, where $b$ is the maximum number of bits used to encode one of the input source probabilities.

Each iterative step of  \cite{yamamoto2015almost}'s  algorithm was  originally implemented using an exponential time Integer Linear Program.  This was later improved by  \cite{dp}  to  $O(n^5)$ time, using  Dynamic Programming (DP) to replace the ILP.  $n$ is the number of different characters in the original souce.

The purpose of this paper is to show how the DP method   can be sped up to $O(n^3)$ time. Combined with \cite{golin2019polynomial},  this yields a $O(n^3 b)$ time algorithm for constructing AIFV-$2$ codes.

Historically, there have been two  major approaches to speeding up DPs. The first is  the Knuth-Yao Quadrangle-Inequality method \cite{Knuth1971,Wachs1989,Yao1980,Yao1982}. The second is the use of ``monotonicity'' or the ``Monge Property'' and   the  application of  the SMAWK \cite{Aggarwal1987} algorithm \cite[Section 3.8]{burkard1996perspectives} (\cite{Woeginger2000} provides a good example of this approach).  There are also variations, e.g., \cite{Eppstein1988}, that while not exactly one or the other, share many of their properties.
\cite{Du2013} provides a recent  overview of the techniques available.

Both methods  improve running times by ``grouping'' calculations.  More specifically,   they all essentially fill in  a DP  table of size $\Theta(n^k)$,  for some $k$, in which calculating an individual table entry requires $\Theta(n)$ work.  Thus, a-priori,   filling in the table seems to require $\Theta(n^{k+1})$ time.  The speedups  work by grouping the entries in sets of size $\Theta(n)$ and calculating all entries in the group in $\Theta(n)$ time.  The Quadrangle-Inequality approach does this via amortization while the SMAWK approach does this by a transformation into another problem (matrix row-minima calculation).  Both approaches lead to a $\Theta(n)$ speedup, permitting filling in the table in an optimal $\Theta(n^k)$ time.

Both DPs   in  \cite{dp} have  $O(n^3)$ size tables with each entry requiring $\Theta(n^2)$ individual evaluation time, leading to the $O(n^5)$ time algorithms.   The main contribution of this paper is the development of  new grouping techniques that permit speeding up the DPs by $\Theta(n^2),$  decreasing the running times to $O(n^3).$ 

More specifically, the table entries are now  partitioned into $\Theta(n)$ groups,  each containing  $\Theta(n^2)$ entries.   For each group,  a  $\Theta(n) \times \Theta(n)$ sized  rectangular  matrix  $M$ is then built; calculating the value of each table entry in the group is shown to be equivalent to performing  a Two-Dimensional Range Minimum  (2D RMQ)  query   on $M$ (along with $O(1)$ extra work). Known results \cite{yuan2010data} on 2D RMQ queries  imply that %after $O(n^2)$ preprocessing time of $M,$ a RMQ  can be implemented in only $O(1)$ time.  
$O(n^2)$  queries can be inplemented using a total of  $O(n^2)$  time.
 Thus all entries in  each group of size $\Theta(n^2)$ can be evaluated in $O(n^2)$ time, leading to an $O(n^3)$ time algorithm.

To the best of our knowledge this is the first time 2D RMQs have been used for speeding up Dynamic Programming in this fashion,  so this technique might be of independent interest.

\Cref{sec:RMQ} quickly reviews known facts about 2D RMQs.  It also introduces the two specialized versions of RMQs that will be needed and shows that they can be solved even more simply (practically) than standard RMQs.
\Cref{sec:The DP}  is the main result of the paper.  It states (before derivation) the two DPs of interest and then describes the new technique  to reduce their evaluation from $\Theta(n^5)$  to $\Theta(n^3).$ 
The remainder of the paper then provides the backstory.  \Cref {sec:The Codes}
 defines the motivating AIVF-$2$ problem and  the technique for solving it.  Finally,  \Cref{sec:Derivation} describes the derivation of the AIFV-$2$ DPs that were solved in \Cref{sec:The DP}.  We emphasize that  while these  DPs are not {\em exactly} the ones  introduced in \cite{dp} they are very similar and  were derived using the same observations and basic tools (the top-down signature technique of \cite{golin,chan2000dynamic}).   The derivation of these new DPs was necessary, though.  Their slightly different structure  is what permits successfully applying the 2D RMQ technique to them

We conclude by noting that AIFV-$2$ codes were later extended to AIFV-$m$ codes by  \cite{introduceMAIFV}. These replace the pair of coding trees by an $m$-tuple. 
% and can achieve $\frac 1 m $ redundancy if permitted an  $m$  bit decoding delay.
The iterative algorithms for constructing these codes use $O(n^{2m+1})$ time DP algorithms that fill in size $O(n^{m+1})$ DP tables  as subroutines.  An interesting direction for future work is whether it is possible to reduce the running times of  evaluating those DP tables by a factor of  $\Theta(n^m)$  via the use of the corresponding  $m$D RMQ algorithms from \cite{yuan2010data}. This would require  a much better understanding of the structure of those DPs in   \cite{introduceMAIFV}  than currently exist.

%% file: RMQ.tex
%\section{Range Minimum Queries}
%\label{sec:RMQ}
%\input{RMQ.tex}

As, mentioned, the speedup  in evaluating the DPs will result from 
grouping  and  then  using  Range Minimum Queries (RMQs). This section quickly reviews facts about  RMQs  for later use.
\begin{definition}[2D RMQ]
Let $M= (M_{i,j})$ be a given $m \times n$  matrix;  $0 \le i \le m,$  $0 \le j \le n$.   The {\em two-dimensional range minimum query (2D RMQ)}  problem is, for $0 \le a_1 \le a_2 \le n$ and $0 \le b_1 \le b_2 \le n$ to return the value 
$$RMQ(M:a_1,a_2, b_1,b_2) \triangleq  
\min\{M_{i,j} \,:\, a_1 \le i \le a_2,\,  b_1 \le j \le b_2\}.
$$
and indices   $i',j'$,  $a_1 \le i' \le a_2,$  $ b_1 \le j' \le b_2$ such that
$$M_{i',j'} = RMQ(M:a_1,a_2, b_1,b_2).
%\min\{M_{i,j} \,:\, a_1 \le i \le a_2,\,  b_1 \le j \le b_2\}.
$$
\end{definition}

This can  be solved using 
\begin{lemma} [\cite{yuan2010data}]
Let $M= (M_{i,j})$  be a given $m \times n$  matrix;  $0 \le i \le m,$  $0 \le j \le n$. There is an $O(m n)$ time algorithm to preprocess  $M$ that permits 
answering any subsequent 2D RMQ query in $O(1)$ time.
\end{lemma}

While theoretically optimal, the algorithm in \cite{yuan2010data} is quite complicated.  To make the speed up more practical to implement, we note in advance that all of the RMQ queries used later  will  be one of the two following specialized types:  
\begin{definition}
\label{def:RRMQ}
Let $M= (M_{i,j})$ be a given $m \times n$  matrix;  $0 \le i \le m,$  $0 \le j \le n$.  Let   $0 \le a\le m,$  $0 \le b \le n$. See \Cref{fig:RRMQ}. 
\begin{itemize}
\item Define a {\em restricted column query} as 
$$RCQ(M:a,b) \triangleq   RMQ(M:a, m,b,b).$$
\item Define a{ \em  restricted RMQ  query}  as 
$$RRMQ(M:a,b) \triangleq   RMQ(M:a, m,0,b).$$
\end{itemize}
\end{definition}

\begin{figure}
\centerline{\includegraphics[width=2.5in]{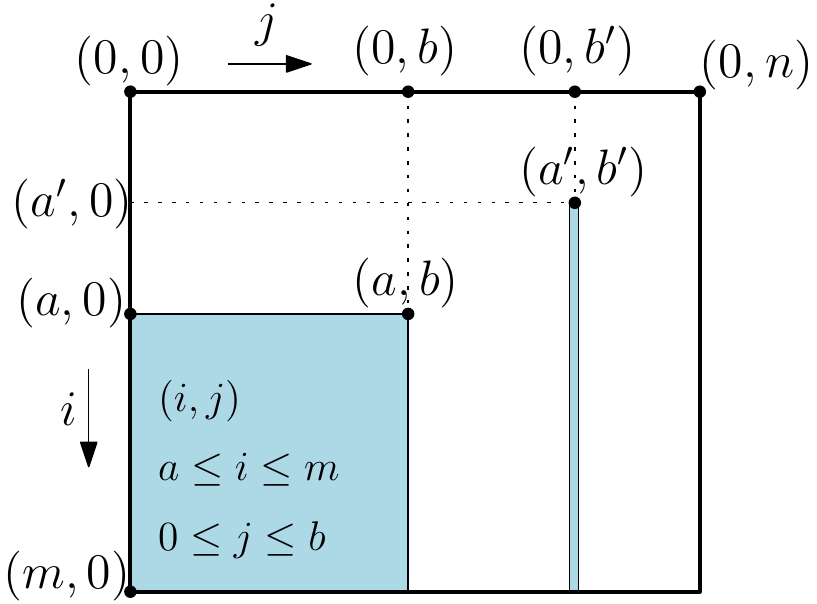}}
\caption{Illustration of \Cref{def:RRMQ}. 
$M$ is an $(m+1) \times (n+1)$ matrix. 
$RCQ(M:a',b')$ is the minimum of the entries in the long thin blue column descending down from $(a',b').$
$RRMQ(M:a,b)$ is the minimum of the entries in the blue rectangle with upper-right corner $(a,b)$.
}\label{fig:RRMQ}
\end{figure}

Directly from the definition,
$$
\forall b,\quad RCQ(M:a,b) = 
\begin{cases}
M_{m,b} & \mbox{if $a=m$},\\
\min\left( M_{a,b}, RCQ(M:a+1,b) \right) &  \mbox{if $a<m$}.
\end{cases}
$$
Thus,  the values of all of the $\Theta(mn)$  possible $RCQ(M:a,b)$ queries (and the associated indices at which minimization occurs) can be easily calculated in $\Theta(mn)$ time.  

Also directly from the definitions, 
$$
 RRMQ(M:a,b) = 
\begin{cases}
RCQ(M:a,0) & \mbox{if $b=0$},\\
\min\left( RRMQ(M:a,b-1), RCQ(M:a,b) \right) &  \mbox{if $b >0$.}\\
\end{cases}
$$
Thus,  assuming that all of the $RCQ(M:a,b)$ have been precalculated,   the values of all of the $\Theta(mn)$  possible $RRMQ(M:a,b)$ queries (and the associated indices at which minimization occurs) can also be easily calculated in $\Theta(mn)$ time.  

For later use we collect this in a lemma.
\begin{lemma} 
\label{lem:RRMQ}
Let $M $ be a given $m \times n$  matrix;  $0 \le i \le m,$  $0 \le j \le n$. There is an $O(m n)$ time algorithm that  calculates the answers to all of the possible
$ RCQ(M:a,b)$ and $RRMQ(M:a,b)$ queries.
\end{lemma}

%% file: DP_and_Speedup_v2.4.tex
%\section{The Dynamic Program and its speedup}
%\label{sec:The DP}
%\input{DP_and_Speedup.tex}

\begin{definition}
\label{def:Wdef}
Let $p_1,\ldots,p_n$ be given such that $\forall i,\, p_i >0$ and $\sum_{i=1}^n p_i =1.$
Set
$$W_m \triangleq  \sum_{j \le  m} p_j,
\quad\mbox{and}\quad
W'_m \triangleq  \sum_{j > m} p_j = 1 - W_m
$$
and, for $ m' < m$,
$$W_{m',m}\triangleq  \sum_{m'  < j \le m} p_j = W_m - W_{m'}.
$$
\end{definition}
The algorithm  precalculates and stores all of the $W_m$ in $O(n)$ time. Subsequently,  the $W_m,$ $W'_m$ and $W_{m',m}$ can all be calculated in $O(1)$ time.

The Dynamic  Programs are defined on $O(n^3)$ size tables that are indexed by {\em Signatures}.  The next two definitions define the Signature set (of indices) and the Dynamic Programming recurrence imposed on them.

\begin{definition}[The Signature Set and costs]\ Let  $C$  ($ 0\le C \le 1$)  be fixed.
\label{def:Exp1}
\begin{itemize}
\item 
Define
$$\CALS_n \triangleq 
\left\{ (m;p;z) \,:\, \mbox{$0 \le z \le m \le n$  and $0 \le p \le n$}
\right\}
$$
to be the {\em signature set} for the problem of size $n.$

\item 
Let  $(m';p';z') \not= (m,p,z) \in \CALS_n$. \\
We say $(m';p';z') $ {\em can be expanded into}  $(m;p;z)$, denoted by
$$(m';p';z')  \rightarrow  (m,p,z),
$$
 if there exists $e_0,e_1$  satisfying 
  \begin{equation}
 \label{eq:pdefe}
 e_0,\,e_1 \ge 0 \quad \mbox{ \rm such that }\quad  0  \le e_0 + e_1 \le p'
 \end{equation} 
 and
  \begin{eqnarray}
 m  &=& m'+e_0+e_1, \label{eq:pdefm}\\
 z   &=& e_1,\label{eq:pdefz}\\
 p &=&  z' + 2(p'-e_0-e_1). \label{eq:pdefp}
 \end{eqnarray}
 \item For $\alpha \in \CALS_n$, define the immediate predecessor set  of $\alpha$ to be
 $$
P(\alpha) \triangleq   \{\alpha' \in  {\cal S}_n \,:\, \alpha' \rightarrow \alpha\}.
$$
\item  Let   $\alpha_1,\alpha_2 \in \CALS_n.$  We say that $\alpha_1$  {\em leads to} $\alpha_2$, denoted by  
$\alpha_1 \rightsquigarrow \alpha_2$,  if there exists a path from $\alpha_1$ to $\alpha_2$ using ``$\rightarrow$''. 
\item Let $I \subset \CALS_n$ and $\alpha \in \CALS_n$.  We say that $I \rightsquigarrow  \alpha$ if $\alpha \not\in I$ and there exists $\alpha' \in I$ such that
$\alpha' \rightsquigarrow  \alpha$.

 \item  Let $\alpha' =(m';p';z')$ and  $\alpha=  (m,p;z)$ where
 $\alpha' \rightarrow \alpha$.  The associated {\em expansion costs} are
 \begin{eqnarray*}
 c_0(\alpha',\,  \alpha) &\triangleq & W'_{m'} + C W_{m'-z',m'},\\
 c_1(\alpha',\, \alpha) &\triangleq & W'_{m'} - C W_{m',m-z}.\\
 \end{eqnarray*}
 \end{itemize}
\end{definition}

The two  dynamic programs  used in the construction of AIFV-$2$ codes are given in the next definition.
\begin{definition}[The $\OPT_s(\alpha)$ tables] 
\label{def:OPT Tables}
\ 
\begin{itemize}
\item 
Let $I_0\subset \CALS_n$ be a given {\em initial  set}  (independent of $n$)  for the $\OPT_0$ table with  known values  $\bar c_0(\alpha)$ for $\alpha \in I_0.$ Now define
$$
 \OPT_0(\alpha)
= 
\begin{cases}
\bar c_0(\alpha)  & \mbox{if  $\alpha \in I_0$}\\
%\min_{\alpha'\,:\, \alpha'\rightarrow  \alpha} 
\min_{\alpha' \in P(\alpha)} 
 \{\OPT_0(\alpha') + c_0(\alpha',\alpha)\} & \mbox{if $I_0  \rightsquigarrow   \alpha$}\\
\infty  & \mbox{otherwise}
\end{cases}
$$
\item 
Let $I_1 \subset \CALS_n$ be a given {\em initial  set}  (independent of $n$)  for the $\OPT_1$ table with  known values  $\bar c_1(\alpha)$ for $\alpha \in I_1.$ Now define
$$
 \OPT_1(\alpha)
= 
\begin{cases}
\bar c_1(\alpha)  & \mbox{if  $\alpha \in I_1$}\\
%\min_{\alpha'\,:\, \alpha'\rightarrow  \alpha} 
\min_{\alpha' \in P(\alpha)} 
 \{\OPT_1(\alpha') + c_1(\alpha',\alpha)\} & \mbox{if $I_1  \rightsquigarrow   \alpha$}\\
\infty  & \mbox{otherwise}
\end{cases}
$$
\item  Furthermore,  for  $s\in \{0,1\},$ for $\alpha \not\in I_s$   with $I_s  \rightsquigarrow   \alpha$, set 
$$
\Pred_s(\alpha) \triangleq  \argmin_{\alpha' \in P(\alpha)} 
\{\OPT_s(\alpha') + c_s(\alpha',\alpha)\}  
$$
\end{itemize}
The $\bar c_s(\alpha)$ for $\alpha \in I_s$ are the initial conditions for the corresponding dynamic programs.  %Note that the values of these $\bar c_s(\alpha)$ may be  dependent upon $n.$
\end{definition}

For intuition, 
let $G_s(n)$ be the directed graph with vertices $\alpha\in\CALS_n$   with the cost of  edge $(\alpha',\,  \alpha)$  being  the expansion cost
 $c_s(\alpha',\,  \alpha)$ except  that edges from $(0;0;0)$ to $\alpha \in I_s$  have cost 
$\bar c_s(\alpha)$ and edges that are not expansions have costs set to $\infty.$  Then $ \OPT_s(\alpha)$ is just the cost of the shortest path from $(0;0;0)$ to $\alpha$ in $G_s(n).$  The actual path could be found by following the $\Pred_s(\alpha)$ pointers backward from $\alpha$. By definition,  the expansion costs $c_s(\alpha',\,  \alpha)$ are all  non-negative, so the $\OPT_s(\alpha)$ values are all well-defined.

The next set of lemmas will imply that  $G_s(n)$ is a Directed Acyclic Graph so the recurrences define a Dynamic Program.
They will also suggest an efficient grouping mechanism, leading to fast evaluation.

\begin{lemma}
\label{lem:2Deqiv}
Let 
$ (m';p';z'), (m;z;p) \in \mathcal{S}_n.$  Then 
$$(m';p';z')  \rightarrow  (m,z,p)$$
if and only if all of 
\begin{eqnarray}
2m' + 2p' + z'  &=& 2 m +p, \label{eq:2Deqiva}\\
m'+p' &\ge& m, \label{eq:2Deqivb}\\
m' &\le& m-z, \label{eq:2Deqivc}\\
(p',z') &\not=& (0,0),  \label{eq:2Deqivd}
\end{eqnarray}
are satisfied.
\end{lemma}
\begin{proof}
First assume that $(m';p';z')  \rightarrow  (m,p,z)$. 

Let $e_0,e_1$ be the  unique pair that satisfies 
(\ref{eq:pdefe})-(\ref{eq:pdefp}). 
Then (\ref{eq:2Deqiva}) follows from
\begin{eqnarray*}
2m' + 2p'   &=&  2(m - e_0-e_1) + (p -z' + 2 e_0 +2 e_1)\\
			&=& 2m + p - z';
\end{eqnarray*}
(\ref{eq:2Deqivb})  follows from 
$$
m'+p' \ge m' + e_0 + e_1 = m;
$$
(\ref{eq:2Deqivc})  follows from 
$$
m- z = m- e_1 = m'-e_0 \ge m'.
$$
(\ref{eq:2Deqivd})  follows from the fact that the combination of $(p',z') =  (0,0)$  and \Cref{eq:pdefp} would imply $p = - (e_0 + e_1)$. Since $ p\ge 0$, this further  implies  $e_0=e_1 =0$  and thus  $m = m'$ and  $p=z=0$.
This would contradict   $(m';z';p') \not= (m;z;p).$
\medskip

For the other direction assume that Equations (\ref{eq:2Deqiva})-(\ref{eq:2Deqivd})  all hold. We
will show that % $e_0=m-m'-z$ and  $e_1=z$ satsify  $e_0,e_1 \ge 0$, $e_0 +e_1 \le p'$ and 
Equations 
(\ref{eq:pdefe})-(\ref{eq:pdefp}) with   $(m';z';p') \not= (m;z;p)$  also all  hold with $e_0=m-m'-z$ and  $e_1=z$.
Equations (\ref{eq:pdefm}) and (\ref{eq:pdefz}) are trivially satisfied.  (\ref{eq:pdefp}) follows from
\begin{eqnarray*}
p &=&  2m'+ 2p'+z' - 2m\\
&=& z'  - 2(m-m') + 2p'\\
&=& z' - 2(e_0 + e_1) + 2p'\\
&=& z' + 2(p'-e_0-e_1).
\end{eqnarray*}

Next note that $e_1 = z \ge 0$ and, from   (\ref{eq:2Deqiva}) and  (\ref{eq:2Deqivc}),
$e_0 = m-z-m' \ge 0$. Finally, from from (\ref{eq:2Deqivb}),  
$p' \ge m-m' = e_0+e_1$ so     
\Cref{eq:pdefe} holds. 

 It only remains to show that  $(m';z';p') \not= (m;z;p).$ Suppose, not and $(m';z';p') = (m;z;p).$
 Then from (\ref{eq:pdefp}),  $e_0=e_1=0$ so from (\ref{eq:pdefz}) $z'=z=0$ and thus from (\ref{eq:pdefp}), $p = 2 p'$ implying $p'=p=0.$  But this contradicts (\ref{eq:2Deqivd}).
%Equation (\ref{eq:2Deqiva})  
\end{proof}

\begin{definition}
\label{def:Idef}
For $d \ge 0$, define
\begin{eqnarray*}
\CALI(d) &\triangleq & \{(m;p;z) \in {\cal S}_n \,:\, 2m + p =d\},\\
\CALI'(d) &\triangleq & \{(m';p';z') \in {\cal S}_n \,:\, 2m' + 2p'  + z'=d \ \mbox{and} \ (p',z') \not=(0,0)\}.
\end{eqnarray*}
%For $\alpha \in \CALI(d)$ 
\end{definition}

Now note that Lemma \ref{lem:2Deqiv} can be rewritten as 
\begin{corollary}  
\label{cor:batch}
If $\alpha \in \CALI(d)$ then 
$$P(\alpha) =
\Bigl\{ (m';p';z') \in \CALI'(d)) \,:\, m'+p' \ge m \mbox{ and } m' \le m-z
\Bigr\} \subseteq \CALI'(d)).$$
\end{corollary}

Next note 
\begin{lemma}
\label{lem:Inc}
 Let $d >0.$  Then
\begin{equation}
\label{eq:inclusion}
 \CALI'(d) \subseteq  \bigcup_{d' < d}\CALI(d').
 \end{equation}
\end{lemma}
\begin{proof}
Let $\alpha =(m',p',z')  \in \CALI'(d).$ Since the $\CALI(d')$ partition  ${\cal S}_n$, there must exist some $d'$ such that $\alpha \in \CALI(d').$  Suppose that $d \le d'.$  Then
$$ 2m'+2p'  +z' = d \le d' = 2m' + p', $$
implying  $p' + z' \le 0$ so $(p,z) = (0,0)$, contradicting the definition of $\CALI'(d)).$ Thus $d' < d.$
Since this is true for all $\alpha  \in \CALI'(d),$
\Cref{eq:inclusion} follows.
\end{proof}

\Cref{cor:batch} and \Cref{lem:Inc}
together imply  that the $\OPT_s(\alpha)$ tables can be evaluated in the order  
   $\alpha \in \CALI(d)$  for  $d=1,2,3\ldots$. 
 This ordering guarantees that when  $\OPT_s(\alpha)$ is being  calculated,  all of the $\OPT_s(\alpha')$ entries for which 
 $\alpha' \in P(\alpha)$ have  been previously calculated.

For many $\alpha,$  $|P(\alpha)| = \Theta(n^2)$,  so calculating
 $\OPT_s(\alpha)$  would require $\Theta(n^2)$ time.  Since 
 $|\CALS_n| = \Theta(n^3)$,  this would imply an $O(n^5)$ time algorithm for filling in the entire table. This is similar to the $O(n^5)$ derivation in   \cite{dp}. We now show how to reduce this down to $O(n^3)$ using RMQs and  \Cref{lem:RRMQ}.

The sped up  $O(n^3)$ algorithm   works in batched stages.  In stage $d,$  the algorithm calculates
$\OPT_s(\alpha)$ for all $\alpha \in \CALI(d)$.  It  first spends $O(n^2)$ time building an associated matrix $M^{d}$ and then reduces the calculation of each $\OPT_s(\alpha)$ to a 2D-RMQ query (and possibly $O(1)$ extra work).

Before starting we quickly note  a small technical issue concerning the DP initial conditions.  Let 
$$\bar d_s = \max_{\alpha = (m;p;z) \in \CALI_s} 2m+p.$$
The starting stage of the algorithms  is just to calculate $\OPT_s(\alpha)$ for all $\alpha \in \CALI(d)$ with $d=1,\ldots,\bar d_s.$  Calculating all of these  requires only  $O(1)$ time.

 We now first describe the complete solution for $\OPT_0$, which will be easier, and then discuss the modifications needed for $\OPT_1.$

Assume then that, for some $d > \bar d_0,$    $\OPT_0(\alpha')$ is already known for  all $\alpha' \in \CALI(d')$, where  $d' < d.$
If  $\alpha=(m;p;z) \in \CALI(d)$ then, by definition, 
\begin{equation}
\label{eq:OPTMINDEF}
 \OPT_0(\alpha) =  \min_ { \alpha'= (m';p',z') \in P(\alpha)}
  \{\OPT_0(\alpha') + CW_{m'-z',m'}\}
\end{equation}
where all the $\OPT_0(\alpha') $ for $\alpha' \in P(\alpha)$ are already known.

Recall that there are $O(n^2)$ signatures  $\alpha' = (m';p';z') \in \CALI'(d)$.  The idea is to  arrange the  corresponding $O(n^2)$ values $\OPT(m';p';z')+ CW_{m'-z,z}$ in an array $M^{(d)}_{i,j}$ in such a way that, for each individual
$\alpha \in I(d)$, the minimization in \Cref{eq:OPTMINDEF} could be performed using just one   2D RMQ query in $M^{(d)}_{i,j}$.

The arrangement will use the invertible transformation (see \Cref{fig:transformation})
\begin{figure*}[t]
 \center
\includegraphics[width=5in]{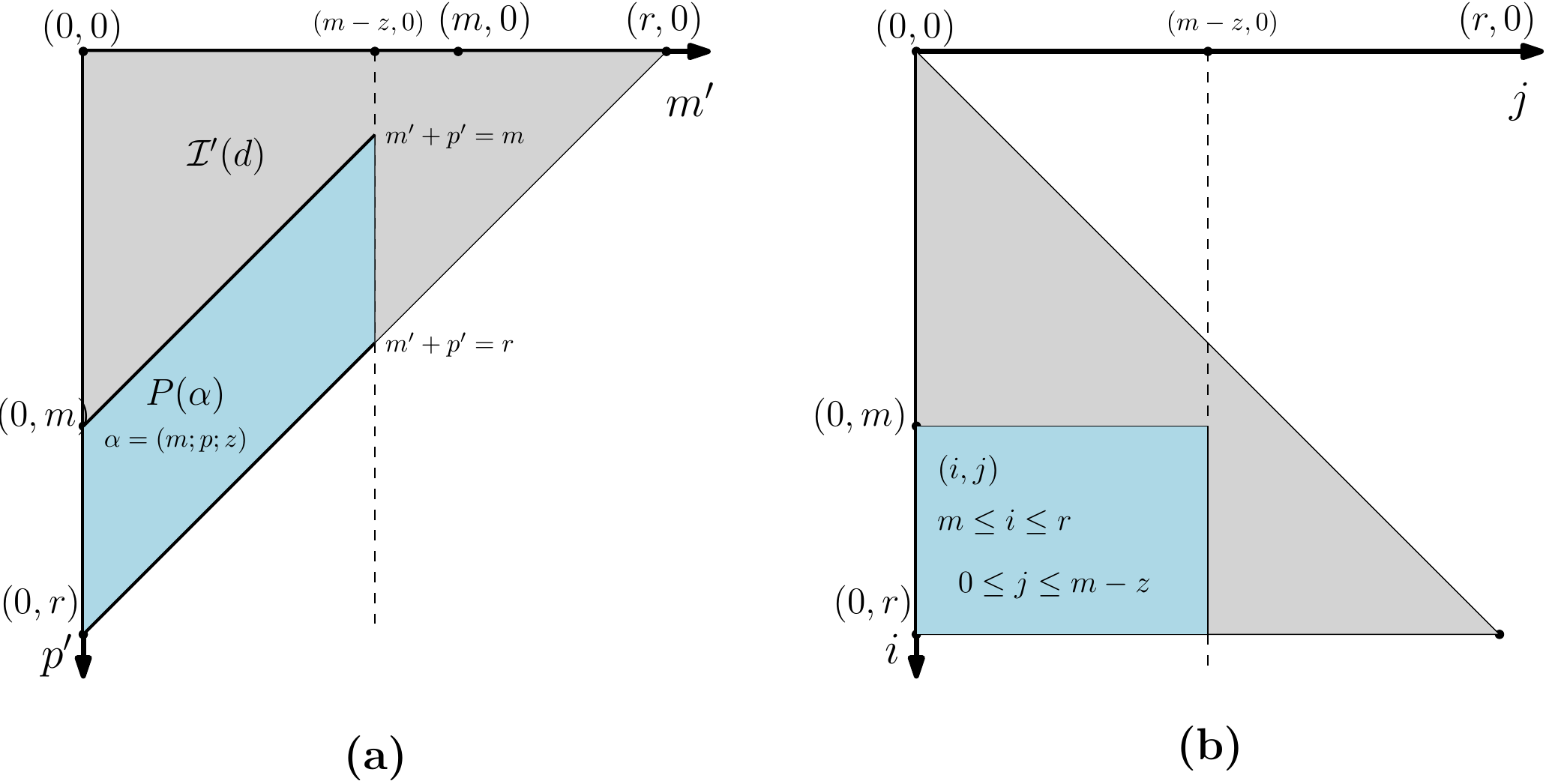}
\caption{The transformation from $(m',p')$ to $(j,i)$  described in the text. From \Cref{def:Idef}, if  $(m',p',z') \in \mathcal{I}'(d)$ then $z'= d - 2 m' - 2 p'$ is uniquely determined by ($m',p').$  In  (a), the right triangle bounded by vertices  $(0,0)$,  $(0,r)$ and $(r,0)$ with $r = \lfloor d/2 \rfloor$ is the  location of all $(p',m')$ pairs such that 
 $(m',p',z') \in \mathcal{I}'(d)$.  The blue shaded parallelogram  is the location of  all $(p',m')$ pairs such that 
 $(m',p',z') \in P(\alpha)$ for some $\alpha = (m,p,z) \in \mathcal{I}(d).$
 (b) illustrates the transformation  $(j,i) = (m',m'+p').$  Note how the blue parallelogram becomes  a rectangle, permitting the use of a 2D RRMQ query.}
\label{fig:transformation}
\end{figure*}
$$j =m'
\quad\mbox{and}
\quad
i= m'+p'.
$$  
Trivially  $j \le i.$  Furthermore,  %if $ (m';p';z') \in \CALI'(d)$ then  
$$(m';p';z') \in \CALI'(d)  \quad \Rightarrow \quad  d = 2 m' + 2p' + z' = 2i + z' $$
which in turn implies
$$2i \le d  \quad \mbox{and} \quad  z'= d - 2 i.
$$
Set $r = \lfloor d/2 \rfloor.$ Then
 %if   $ (m';p';z') \in \CALI'(d)$ then
\begin{equation}
\label{eq:trans1}
(m';p';z') \in \CALI'(d)  \quad \Rightarrow \quad 0 \le  j \le i \le r
%\quad 2i \le d,
\quad 
\mbox{and} \quad (m';p';z') = (j;i-j;d-2i).
\end{equation}
Furthermore,  working backwards,
\begin{equation}
\label{eq:trans12}
0 \le  j \le i \le r 
\quad  \mbox{and} \quad (i-j;d-2i) \not= (0,0)
\quad \Rightarrow \quad  (j;i-j;d-2i) \in \CALI'(d).
\end{equation}
where the second condition comes from the fact that
 $(m',p,z') \not\in \mathcal{I}'(d)$  if  $(p',z') = (0,0)$.

 The preceding discussion motivates defining the
$(r+1) \times (r+1)$
matrix (indices of $i$ and $j$ start at $0$) 
$$
\small % \hspace*{-.2in}
M^{(d)}_{i,j} \triangleq 
\begin{cases}
\infty & \mbox{if $i> j$}\\
\infty &  \mbox{if $i =j = \frac d 2$}\\
 \OPT_0(j;i-j; d-2i) + C W_{j-(d-2i),j}  & \mbox{Otherwise}.
\end{cases}
$$
Since all the values referenced are already known, this matrix can be built in $O(r^2) = O(n^2)$ time.

Then, if $\alpha \in \CALI(d),$ from \Cref{cor:batch},
\begin{eqnarray*}
 \OPT_0(\alpha)
 &=& 
% \min_ {\substack{ \alpha'= (m';p',q')\\ \alpha'  \in P(\alpha)}}
  \min_ { \alpha'= (m';p',z')  \in P(\alpha)}
  \{\OPT_0(\alpha') + CW_{m'-z',m'}\}\\
  &=&  \min\{ M^{(d)}_{i,j} \,:\,  i =m'+p' \ge  m  \mbox{ and } j  = m'\le  m-z\}\\
  &=&  RMQ\left(M^{(d)}_{i,j}: m,r,0,m-z\right)\\
  &=&   RRMQ\left(M^{(d)}_{i,j}: m,m-z\right) 
\end{eqnarray*}
Note that the RRMQ query result also provides the {\em indices} of the minimizing entry, which provides the corresponding $\Pred_0(\alpha)$ value as well.

\Cref{lem:RRMQ} permits calculating {\em all} the $O(r^2)$ $RRMQ\left(M^{(d)}_{i,j}: a,b\right)$ values in $O(r^2) = O(n^2)$ time.
Thus, all of the  $ \OPT_0(\alpha)$ for $\alpha \in \CALI(d)$ (and  their corresponding $\Pred_0(\alpha)$ values)  can be calculated in $O(n^2)$ total time.
 Doing this for all $O(n)$ values of $d> \bar d_0$ in increasing order, yields the required  $O(n^3)$ time algorithm for filling in the $\OPT_0$ matrix.

\medskip

We next describe the more complicated algorithm for the $\OPT_1$ case.

Assume that  $\OPT_1(\alpha)$ is already known for  all $\alpha \in \CALI'(d)$, $d' < d.$
If  $\alpha=(m;p;q) \in \CALI(d)$ then,  similar to the $\OPT_0$ case,
\begin{eqnarray*}
 \OPT_1(\alpha)
%&=&
% \min_ {\substack{ \alpha'= (m';p',q')\\ \alpha'  \in \CALI'(d)\\   \alpha' \rightarrow \alpha}}
%  \{\OPT^1(\alpha') - CW_{m',m-z}\}\\
   &=& 
% \min_ {\substack{ \alpha'= (m';p',q') \\ \alpha' \in P(\alpha)}}
   \min_ { \alpha'= (m';p',z')  \in P(\alpha)}
  \{\OPT_1(\alpha')  - CW_{m',m-z}\}
\end{eqnarray*}
where all the $\OPT_1(\alpha') $ for $\alpha' \in P(\alpha)$ are already known.

Following the approach in the $\OPT_0$ algorithm, for fixed $d,$ 
we would like  to arrange the  $O(n^2)$ values
$\left(\OPT_1(\alpha')  - CW_{m',m-z}\right)$ for
 $\alpha' = (m';p';z') \in \CALI'(d)$,
  appropriately in an array so that each $\OPT_1(\alpha)$ entry could be resolved using one 2D RMQ query.
The difficulty  is that the {\em values} of the array entries  depend upon {\em  both}  $\alpha$ and  $\alpha'$.  More specifically, the $CW_{m',m-z}$ term would have to be reprocessed  for each $(m,z)$ pair.  Thus, no fixed $M_{i,j}$ array, independent of $(m,z)$, could be defined.

Instead, we utilize a relationship between different queries.  More specifically, 
let  $\alpha=(m;p;z) \in \CALI(d)$.  
From Equation (\ref{eq:2Deqivc}), 
 $z\le m-m' \le m.$   If $z =m,$ then $m'=0$ so   $W_{m',m-z} = W_{0,0} =0$ and
\begin{eqnarray*}
 \OPT_1(\alpha)
   &=& 
    \min_ { \alpha'= (m';p',z')  \in P(\alpha)}
  \{\OPT_1(\alpha')  - CW_{m',m-z}\}\\
  &=&
      \min_ { \alpha'= (m';p',z')  \in P(\alpha)}
  \{\OPT_1(\alpha')\}.
\end{eqnarray*}
If $z < m$ then, splitting into the cases $m' = m-z$ and  $m' \le  m-z -1$ yields, 
$$\OPT_1(\alpha) = \min(A, B)$$
where 
\begin{eqnarray*}
A
    &\triangleq & 
 \min_ {\substack{ \alpha'= (m';p',z') \in P(\alpha)\\  m' =m-z}}
  \{\OPT_1(\alpha')  - CW_{m',m-z}\},\\
B
    &\triangleq & 
 \min_ {\substack{ \alpha'= (m';p',z') \in P(\alpha)\\  m' \le m-z-1}}
  \{\OPT_1(\alpha')  - CW_{m',m-z}\}.
\end{eqnarray*}
First note that if $m' = m-z$, then $W_{m',m-z} =0$ so
\begin{eqnarray*}
A
%    &=& 
% \min_ {\substack{ \alpha'= (m';p',z') \in \CALI'(d) \\  m'+p' \ge m\\  m' =m-z}}
%  \{\OPT_1(\alpha')  - CW_{m',m-z}\}\\
     &=& 
 \min_ {\substack{ \alpha'= (m';p',z') \in \CALI'(d) \\  m'+p' \ge m\\  m'  = m-z}}
  \{\OPT_1(\alpha')\}.
\end{eqnarray*}
Next note that, from \Cref{cor:batch},
\begin{equation*}
\begin{aligned}
P((m;p;z))  \cap &  \{(m';p';z') \,:\, m' \le m - z - 1\}\\
=&\Bigl\{ (m';p';z') \in \CALI'(2m+p)) \,:\, m'+p' \ge m \mbox{ and } m' \le m-z-1\Bigr\}\\
=& P((m; p; z+1))
\end{aligned}
\end{equation*}
and from \Cref{def:Wdef}
$$ W_{m',m-z} = p_{m-z} + W_{m',m-(z+1)}.$$
Thus
\begin{eqnarray*}
B
%    &=& 
% \min_ {\substack{ \alpha'= (m';p',z') \in \CALI'(d) \\  m'+p' \ge m\\  m' \le m-z-1}}
%  \{\OPT_1(\alpha')  - CW_{m',m-z}\}\\
   &=& 
 -C p_{m-z} +
% && \min_ {\substack{ \alpha'= (m';p',z') \in \CALI'(d) \\  m'+p' \ge m\\  m' \le m-(z+1)}}
%  \{\OPT_1(\alpha')  - CW_{m',m-(z+1)}\}\\
 %\min_ {\substack{ \alpha'= (m';p',z')\\ \alpha'  \in P(m;p;z+1)}}
 \min_ { \alpha'= (m';p',z') \in P((m;p;z+1))}
  \bigl\{\OPT_1(\alpha')  - CW_{m',m-(z+1)}\bigr\}\\
  &=&  \OPT_1(m;p;z+1) - C p_{m-z}
\end{eqnarray*}

Again use the same transformation  $j = m'$ and $i = m'+p'$  so that 
\Cref{eq:trans1,eq:trans12} apply.
%  if   $ (m';p';z') \in \CALI'(d)$ then
%$$
%j \le i,\quad 2i \le d,\quad 
%(m';p';z') = (j;i-j;d-2i).
%$$
Set $r = \lfloor d/2 \rfloor$,   define the $(r+1) \times (r+1)$  array
$$
M^{[d]}_{i,j} \triangleq 
\begin{cases}
\infty & \mbox{if $i < j$}\\
\infty &  \mbox{if $i = j = \frac d 2$}\\
 \OPT_1(j;i-j; d-2i)   & \mbox{Otherwise}.
\end{cases}
$$
Next,  use 
\Cref{lem:RRMQ} to  calculate {\em all} the $O(r^2)$ $RCQ\left(M^{(d)}_{i,j}: a,b\right)$ values in $O(r^2) = O(n^2)$ time.

Let  $\alpha=(m;p;z) \in \CALI(d)$. Then, from the discussion above,
\begin{itemize}
\item[If $m=z$,]  
\begin{eqnarray*}
 \OPT_1(\alpha)
     &=& 
 \min_ {\substack{ \alpha'= (m';p',z') \in \CALI'(d) \\  m'+p' \ge m\\  m'  = 0}}
  \{\OPT^1(\alpha')\}\\
    &=&  \min\left\{ M^{[d]}_{i,j} \,:\,  i \ge  m  \mbox{ and } j  = 0\right\}\\
  &=&  RMQ\left(M^{[d]}_{i,j}: m,r,0,0\right)\\
  &=& RCQ\left(M^{[d]}_{i,j}: m,0\right)
\end{eqnarray*}
which is already known.
\item[If $z < m$,] 
\begin{eqnarray*}
A
    &=& 
 \min_ {\substack{ \alpha'= (m';p',z') \in \CALI'(d) \\  m'+p' \ge m\\  m'  = m-z}}
  \{\OPT_1(\alpha')\}\\
    &=&  \min\left\{ M^{[d]}_{i,j} \,:\,  i \ge  m  \mbox{ and } j  = m-z\right\}\\
  &=&  RMQ\left(M^{[d]}_{i,j}:m,r,m-z,m-z\right)\\
  &=& RCQ\left(M^{[d]}_{i,j}:m,m-z\right).
\end{eqnarray*}
Thus, for $\alpha=(m;p;z)$ with $ z < m$,
\begin{eqnarray}
 \OPT_1(\alpha)
     &=& \min(A,B) \label{eq:OPT2rec}\\
     &=& \min\Bigl(RCQ\left(M^{[d]}_{i,j}:m,m-z\right),\, \OPT_1(m;p;z+1) - C p_{m-z}\Bigr) \nonumber
\end{eqnarray}
which, since $RCQ\left(M^{[d]}_{i,j}:m,m-z\right)$ is already known,
 can be calculated in $O(1)$ time if $\OPT_1(m;p;z+1)$ had already been calculated. The associated $\Pred_1(\alpha)$ can be found appropriately.
\end{itemize}

%ote that the RMQ query result also provides the {\em indices} of the minimizing entry, which provides the corresponding $\Pred_0(\alpha)$ value as well..

This permits calculating $\OPT_1(\alpha)$ and (and  their corresponding $\Pred_1(\alpha)$ values)   for all $\alpha = (m;p;z) \in \CALI(d)$ in a total of $O(n^2)$ time as follows: 
\begin{enumerate}
\item First spend $O(n^2)$ time
calculating  {\em all} the  $RCQ\left(M^{(d)}_{i,j}: a,b\right)$ values.
\item  For each of the $O(n)$ possible  fixed pairs $m,p$ satisfying $2m+p =d$
\begin{enumerate}
\item  Set  $ \OPT_1(m;p;m)=RRMQ\left(M^{[d]}_{i,j}: m,0\right)$.
\item Then, for $z=m-1,m-2,\ldots,$ calculate $\OPT_1(m;p;z)$ in $(1)$ time from  $\OPT_1(m;p;z+1)$ using Equation (\ref{eq:OPT2rec}). %  and a 2D RMQ query  in $O(1)$ time.
\end{enumerate}
\end{enumerate}

Since this is $O(n^2)$ time for fixed $d,$  
 doing this for all $O(n)$ values of $d> \bar d_1$ in increasing order yields the required  $O(n^3)$ time algorithm for filling in the $\OPT_1$ matrix.

%% file: AIFV_Intro.tex
%\section{A Quick Introdution to AIFV-$2$ codes}
%\label{sec:The Codes}
%\input{AIFV_Intro.tex}

{\em Note:  This introduction is copied with some small modifications, from \cite{golin2020polynomial}.}

Let $X$ be a memoryless  source over a finite alphabet $\mathcal{X}$ of size $n$.
$ \forall a_i \in \mathcal{X}$,   let $p_i = P_{X}(a_i)$ denote the probability of $a_i$ ocurring.
Without loss of generality  we assume that %$p_i \ge p_j,$ % $ P_X(a_i)\geq P_X(a_j)$ 
%for $1\leq i \leq j \leq n$. 
%(Other wise just sort the alphabets according to the probabilities and rename them). We denote $p_i=P_X(a_i), \forall i$. 
%Then 
$$p_1 \geq p_2 \geq \cdots \geq p_n >0  \text{ and } \sum_{i=1}^{n}p_i = 1.$$
A  {\em codeword} $c$ of a binary AIFV code is  a string in $\{0,1\}^*.$
%$ \mathcal{Y}^*$ where   $\mathcal{Y}=\{0,1\}$. 
$|c|$ will denote   the length of codeword $c$.

\begin{figure*}[t]
 \center
\includegraphics[width=4.5in]{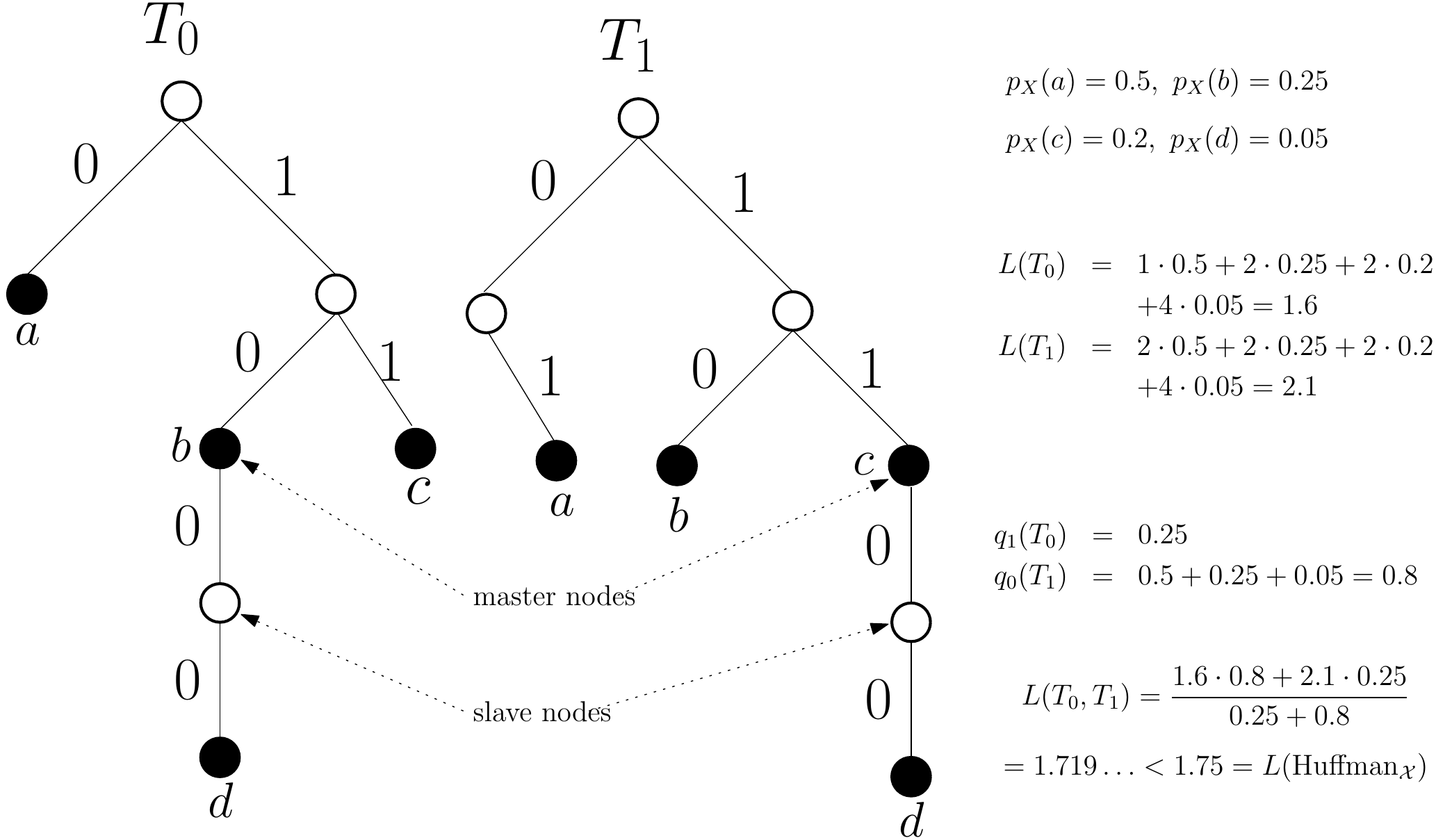}
  \caption{A binary AIFV-$2$ code for  $\mathcal{X}=\{a,b,c,d\}$ with associated probabilities. The encoding of  ${\bf b}d {\bf b}c {\bf a}a$ is
  ${\bf Y} = {\bf 10}1100 {\bf 10} 11 {\bf 01}0 .$ Note that $d,c$ and the first $a$  were encoded using $T_1$ while  the other letters were encoded using $T_0.$ This code has cost $\approx 1.72$ which is better than the optimal Huffman code for the same source which has 
 $L(\rm{Huffman_{\mathcal{X}}})=1.75$ 
  }
  \label{fig:BBB}\end{figure*}

%We assume that all probabilities are represented  using binary. Let $b_i$ be the number of bits used to represent $p_i$.
%For later use we set $b= \max_i b_i.$ 
%Since $\min_i p_i =  p_n \le \frac 1 n$, this implies $b = \Omega(\log n)$.  Note that that the total size of the problem's input is
%$L = \sum_i b_i = \Omega(n \log n).$
We now briefly describe the structure of  Binary AIFV-$2$ codes using the terminology of \cite{introduceMAIFV}.
See  \cite{introduceMAIFV} for more details and \Cref{fig:BBB} for an example.

Codes are  represented via binary trees with left edges  labelled by ``$0$'' and right edges by ``$1$''.
A Binary AIFV-$2$ code  is {\em a pair}  of binary code trees,  $T_0, T_1$ satisfying: 
 \begin{itemize}
 \item Complete internal nodes in $T_0$ and $T_1$ have both left and right children.
\item Incomplete internal nodes (with the unique exception of the left child of the root of $T_1$)  have only a ``$0$'' (left) child.\\
Incomplete internal nodes are labelled as  either  {\em master} nodes or  {\em slave} nodes. 
\item  A master node  must be an incomplete node with an incomplete child\\
The child of a master node is a slave node.\\
{\em This implies that a master node is connected to its unique grandchild via  ``$00$'' with the intermediate node being  a slave node.}
%\vspace*{-.05in}
\item Each source symbol is  assigned to one node in $T_0$ and one node in $T_1$. \\The nodes to which they are assigned are either leaves  or master nodes.\\
{\em Symbols are not assigned to complete internal nodes or slave nodes.}
\item The root of $T_1$ is complete and its ``$0$'' child is a slave node.\\
%must have two children connected with code symbols '0', and '1'. 
The  root of $T_1$ has no  ``$00$'' grandchild.

\end{itemize}

Let $c_s(a), s\in \{0,1\}$ denote the codeword of $a\in \mathcal{X}$ encoded by $T_s$. 
The encoding procedure for a sequence  $x_1, x_2\ldots$ of source symbols works as follows.\\
0.  Set $s_1 = 0$ and $j=1.$\\
1. Encode $x_j$ as $c_{s_j}(x_j).$\\
2. If $c_{s_j}(x_j)$ is a leaf  in $T_{s_j}$, then set $s_{j+1} =0$ \\ \hspace*{.4in} else   set 
$s_{j+1} =1$   \quad \%\ this occurs when $c_{s_j}(x_j)$ is a master node  in $T_{s_j}$ \\
3. Set $j=j+1$ and Goto 1.

\medskip

Note that a symbol is encoded using $T_0$ if and only if its predecessor was encoded using  a leaf node and it is encoded 
using $T_1$ if and only if its predecessor was encoded using  a master node. %sometimes flips  the encoding procedure between the two trees.  
The decoding procedure is  a straightforward reversal of the encoding procedure. Details are provided in 
\cite{original} and \cite{dp}.  The important observation is that identifying the end of  a codeword might first require reading an extra two bits past its ending, resulting in a two bit delay,  so decoding  is not instantaneous.

Following \cite{original}, we can 
now  derive the average codeword length of a binary AIFV-$2$ code  defined by trees $T_0, T_1$.  The average codeword length $L(T_s)$ of $T_s$, $s \in \{0,1\},$  is 
$$L(T_s)=\sum_{i=1}^{n}|c_s(a_i)|p_i.$$

If the current symbol $x_j$ is encoded by a leaf (resp. a master node) of $T_{s_j}$, then the next symbol $x_{j+1}$ is encoded by $T_0$ (resp. $T_1$). This process can be modelled as  a two-state Markov chain with the state being the current encoding tree. 
Denote the transition probabilities for switching from code tree $T_s$ to $T_{s'}$ by $q_{s'}(T_s)$. Then, from the definition of the code trees and the encoding/decoding protocols:
$$q_0(T_s)=\sum_{a\in \mathcal{L}_{T_s}}P_X(a)
\quad \mbox{and} \quad 
q_1(T_s)=\sum_{a\in \mathcal{M}_{T_s}}P_X(a)$$
where $\mathcal{L}_{T_s}$ (resp. $\mathcal{M}_{T_s})$ denotes  the set of source symbols $a\in \mathcal{X}$ that are  assigned to a leaf node (resp. a master node) in $T_s$.

Given binary AIFV-$2$ code $T_0,T_1,$ 
as the number of symbols being encoded approaches infinity, the stationary probability of using  code tree $T_s$ can then be calculated to be 
\begin{equation}
\label{eq:2Dnondegen}
P(s|T_0, T_1)=\frac{ q_s(T_{\hat{s}})}{q_1(T_0)+q_0(T_1)}
\end{equation}
where $\hat{s}\in \{0,1\}, s\neq \hat{s}$.

%Given a binary AIFV code defined by trees $T_0, T_1$,  
The average (asymptotically) codeword length (as the number of characters encoded goes to infinity)  
of a binary AIFV-$2$ code
is then
\begin{equation}
\label{eq:LAIFV def}
L_{AIFV}(T_0, T_1)=P(0|T_0, T_1)L(T_0)+P(1|T_0, T_1)L(T_1)
\end{equation}

\begin{algorithm*}[t]
  \caption{Iterative algorithm to construct an optimal binary AIFV-$2$ code \cite{yamamoto2015almost,dp}}\label{DIV}
  \begin{algorithmic}[1]
  \State $m \leftarrow 0;$\  $C^{(0)}=2-\log_2(3)$
  \Repeat
    \State $m\leftarrow m+1$
    \State $T_0^{(m)}=\text{argmin}_{T_0 \in \mathcal{T}_0(n)}\{L(T_0)+C^{(m-1)}q_1(T_0)\}$
    \State $T_1^{(m)}=\text{argmin}_{T_1 \in \mathcal{T}_1(n)}\{L(T_1)-C^{(m-1)}q_0(T_1)\}$
    \State Update  $$C^{(m)}=\frac{L\left(T_1^{(m)}\right)-L\left(T_0^{(m)}\right)}{  q_1\left(T_0^{(m)}\right) + q_0\left(T_1^{(m)}\right)  }$$
  %   \State Update cost as $C^{(m)}=\frac{L(T_1^{(m)})-L(T_0^{(m)})}{  q_1(T_0^{(m)}) + q_0(T_1^{(m)})  }$
  \Until $C^{(m)}=C^{(m-1)}$
  \State //  Set $C^* = C^{(m)}. $ Optimal binary AIFV-$2$ code is $T_0^{(m)}$, $T_1^{(m)}$
  \end{algorithmic}
  \label{alg: alg1}
\end{algorithm*}

\cite{original, yamamoto2015almost} showed that the  binary AIFV-$2$ code $T_0,T_1$ minimizing \Cref{eq:LAIFV def} 
 can be obtained by \Cref{alg: alg1}, in which $\mathcal{T}_0(n)$ (resp. $\mathcal{T}_1(n)$) is the set of all possible $T_0$ (resp. $T_1$) coding trees.
 %in which the minimizations on lines 4,5 are over all binary coding trees. %for $n$ words.
 It implemented the minimization (over all coding trees)  in lines 4 and 5 as an  ILP.
 %integer linear  programming problem.  
 In a  later paper \cite{dp}, the authors replaced this ILP with a  $O(n^5)$ time and $O(n^3)$ space
 DP % Dynamic Program 
 that modified a top-down tree-building DP   from 
\cite{golin,chan2000dynamic}.

\cite{dp,yamamoto2015almost} proved algebraically that \Cref{alg: alg1}  would terminate after a finite number of steps and that the  resulting tree pair $T_0^{(m)},$  $T_1^{(m)}$ is an optimal Binary AIFV-$2$ code. They were unable, though, to provide any bounds on the number of steps needed for termination. \cite {golin2019polynomial} then gave two new iterative algorithms that provably terminated in $O(b)$ iterations, where
 $b$ is the maximum number of bits required to store any of the probabilities $p_i$ (so these were weakly polynomial algorithms). More formally, let $o_i,b_i$ be such that $p_i = o_i 2^{- b_i}$ where $o_i < 2^{b_i}$ is an odd positive integer.  Then $b = \max_i b_i.$
 
 Each iteration step of \cite {golin2019polynomial}'s algorithm ran $O(1)$ of the DPs from  \cite{dp} so its  full algorithm for constructing  optimal AIFV-$2$ codes  ran in $O(n^5 b)$ time.  The results of this paper replace the $O(n^5)$-time DPs with $O(n^3)$-time DPs, leading to $O(n^3 b)$-time algorithms  for constructing  optimal AIFV-$2$ codes.
 
 We conclude this section by noting that the correctness of the DPs defined in both \cite{dp} and the next section assume  that $0 \le C^{(i)} \le 1$.  The need for this assumption  was implicit in \cite{dp} and is  made explicit in \Cref {lem:DPStructure} in the next section. The validity of this assumption was proven in \cite{golin2020polynomial}.

%% file: Derivation_v2.tex
Each iteration step in both  \cite{dp} and \cite{golin2019polynomial} requires finding trees that satisfy 
\begin{equation}
T_0(C) \triangleq  \text{argmin}_{T_0 \in \mathcal{T}_0(n)} \left\{\Cost_0(T_0:C) \right\},  \label{eq:argminT0}
\end{equation}
\begin{equation}
T_1(C) \triangleq   \text{argmin}_{T_1 \in \mathcal{T}_1(n)}  \left\{\Cost_1(T_1:C) \right\} \label{eq:argminT1},
\end{equation}
where
\begin{equation}
\label{eq:Cost0def}
\Cost_0(T:C)  \triangleq  L(T) + C q_1(T),
\end{equation}
\begin{equation}
\label{eq:Cost1def}
\Cost_1(T:C)  \triangleq  L(T) - C q_0(T).
\end{equation}
Since $C$ will be fixed at any iteration stage, we simplify our notation by  assuming $C$ fixed   and writing   $\Cost_0(T)$ and $\Cost_1(T)$ to denote
\Cref{eq:Cost0def,eq:Cost1def}.

\begin{definition}
Let $T$ be  a binary AIFV coding tree.  Define 
$$\forall a_i \in \mathcal{X},\quad
\begin{array}{ccl}
c_T(a_i)  &\triangleq &\mbox{codeword in $T$ associated with $a_i$},\\
d_T(i)    &\triangleq &  |c_T(a_i)|.
\end{array}
 $$
By the natural correspondence,  $d_T(i)$ is the depth of the node in $T$ associated with  $a_i$ so $L(T) = \sum_{i=1}^n d_T(i) p_i.$
% Note that this node might be either a leaf of a master node.  \
Further define
$$\forall a_i \in \mathcal{X}, \quad
m_T(i) \triangleq 
\begin{cases}
1 & \mbox{if $c_T(a_i)$ is a master node in $T$},\\
0 & \mbox{if $c_T(a_i)$ is a leaf in $T$,}
\end{cases}, \quad
\ell_T(i) 
\triangleq 
\begin{cases}
0 & \mbox{if $m_T(i) =1$}.\\
1 & \mbox{if $m_T(i) =0$}.
\end{cases}
$$
$m_T(i)$ and $\ell_T(i)$ are indicator functions as to whether $a_i$ is encoded by a master node or a leaf in $T,$ 
so, $\forall i,$  $m_T(i)+\ell_T(i)=1.$
%
%$$\forall a_i \in \mathcal{X}, \quad
%\ell_T(i) =
%\begin{cases}
%1 & \mbox{if $c_T(a_i)$ is a leaf node in $T$},\\
%0 & \mbox{if $c_T(a_i)$ is a master node  in $T$.}
%\end{cases}
%$$
\end{definition}

Note that using this new notation
\begin{eqnarray*}
\Cost_0(T)  &=& \sum_{i=1}^n d_T(i) p_i + C \sum_{i=1}^n m_T(i) p_i,\\
\Cost_1(T)  &=& \sum_{i=1}^n d_T(i) p_i - C \sum_{i=1}^n \ell_T(i) p_i.
\end{eqnarray*}

We now show that $0 \le C \le 1$ implies that $T_0(C)$ and $T_1(C)$  can be assumed to possess a nice ordered structure.
\begin{lemma}
\label{lem:DPStructure} 
Let $0 \le C \le 1$.
Then,  if $s=0$   (resp. $s=1$) there exists a tree $T_0(C)\in \mathcal{T}_0(n)$  (resp. $T_1(C)\in \mathcal{T}_1(n)$) satisfying \Cref{eq:argminT0} (resp. \Cref{eq:argminT1}) that,  for all $i <j$, satisfies the following two properties:
\begin{enumerate}
%\setcounter{enumi}{1}
%\item   $|c_s(a_i)| \le |c_s(a_j)|.$
\item[(P1)]  $d_{T_s}(i) \le  d_{T_s}(j).$
 \item[(P2)] If $d_{T_s}(i)  =  d_{T_s}(j)$  and   
 $m_{T_s}(i)  = 1$   then      $m_{T_s}(j)=1$. % as well.
\end{enumerate}
\end{lemma}

\begin{proof}
We say that $T_1=T_1(C)$ (resp $T_2 = T_2(C)$)  is  a minimum cost tree (for $s$) if it  satisfies
Equation (\ref{eq:argminT0}) (resp. (\ref{eq:argminT1})).

The proof  follows from  swapping arguments.   ``Swapping''  $i$ and $j$ means assigning the old codeword $c_{T_s}(a_i)$ to $a_j$ and vice-versa. Let $T'_s$ be the tree resulting from swapping $i$ and $j$.  %Note that $T'_s$ 

 The following observation is  a straightforward calculation:
$$\Cost_s(T'_s)  =  \Cost_s(T_s) -\left(d_{T_s}(i) - d_{T_s}(j)\right) (p_i - p_j) + \delta(i,j)$$
where
$$
\delta(i,j) \triangleq 
\begin{cases}
0 & % \mbox{ if $c_s(a_i) \in  L$,  and  $c_s(a_j)\in L$}\\
 \mbox{ if $m_{T_s}(i) = m_{T_s}(j)$},\\
%0 & % \mbox{ if $c_s(a_i) \in  L$,  and  $c_s(a_j)\in L$}\\
% \mbox{ if $\ell_{T_s}(i) = \ell_{T_s}(j)=1$},\\
%0 & % \mbox{ if $c_s(a_i) \in  M$,  and  $c_s(a_j)\in M$}\\
% \mbox{ if $m_{T_s}(i) = m_{T_s}(j)=1$},\\
- C(p_i- p_j) &%  \mbox{ if $c_s(a_i) \in  M$,  and  $c_s(a_j)\in L$}\\
 \mbox{ if $m_{T_s}(i) =1$,  and  $\ell_{T_s}(j) =1$},\\
C(p_i- p_j) &  %\mbox{ if $c_s(a_i) \in  L$,  and  $c_s(a_j)\in M$}\\
\mbox{ if $\ell_{T_s}(i) =1$,  and  $m_{T_s}(j)=1$}.
\end{cases}
$$
We say that $(i,j)$ is an {\em inversion} for $T_s$ if $i < j$ and  $d_{T_s}(i) > d_{T_s}(j)$.

The calculations above and the fact that $0 \le C \le 1$,  immediately imply that if $(i,j)$ is an inversion for $T_s$ then 
$$\Cost_s(T'_s)  \le  \Cost_s(T_s).$$

Now let $T_s$ be a minimum cost tree for $s$ that  has the minimum number of inversions among all such trees.
If no inversion exists, then $T_s$ satisfies (P1). Otherwise,  let 
$(i,j)$ be the inversion that minimizes  $j-i$. Swapping $i$ and $j$ decreases the number of inversions by $1$ while not increasing the cost of the tree, contradicting the definition of $T_s$.  
We may  therefore assume that $T_s$  contains no inversion and satisfies (P1).

%Now let $T_s$ be  atre
%To prove (P2) suppose that $d_{T_s}(i) = d_{T_s}(j)$.     
Now say that $(i,j)$ is an {\em $m\ell$-inversion}  in $T_s$ if   $i < j$, $d_{T_s}(i) = d_{T_s}(j)$,
 $m_{T_s}(i)=1$   and  $\ell_{T_s}(j)=1$.
  Let $T_s$ be a minimum cost  tree for $s$ that satisfies (P1) and has
the fewest number of $m \ell$-inversions. If no $m \ell$-inversion exists, then $T_s$ also satisfies (P2) so the  lemma is correct.  Otherwise
let $(i,j)$ be an $m \ell$-inversion that minimizes $j-i.$
Let $T'_s$ be the tree that results by swapping $i$ and $j.$  Then $T'_s$ will still satisfy (P1) but the numbers of inversions will decrease by $1$ while 
$$ \Cost_s(T'_s)  =   \Cost_s(T_s) - C(p_i-p_j) \le \Cost_s(T_s). 
$$
This contradicts the definition of $T_s.$  We may therefore assume $T_s$  contains no inversions and satisfies both  (P1) and (P2).
\end{proof}

The consequences of  
Lemma \ref{lem:DPStructure} can be seen in \Cref{fig:Sigs}.  The Lemma   implies that the optimization in \Cref{eq:argminT0}  (resp.  \Cref{eq:argminT1}) can be restricted to trees that satisfy Properties  (P1) and (P2).  In particular, the indices of codewords on a level are smaller than the indices of codewords on deeper levels.  Also,  on  any  given level, the indices of the leaves are smaller than the indices of the master nodes. We therefore henceforth assume that all trees in $\mathcal{T}_0(n),$  and $\mathcal{T}_1(n)$ satisfy these properties.

\begin{definition}[Partial Trees and Truncation]  See \Cref {fig:Expand}.
\begin{itemize}
\item 
A {\em partial binary AIFV code tree} ({\em  partial tree} for short) $T$ is one that satisfies all of the  conditions of a binary AIFV code tree and  properties  (P1), (P2) except that it contains $ m \le  n$ codewords.  By (P1), the $m \le n$ codewords it contains  are  $c_T(a_1),\ldots,c_T(a_m).$  

\item For $s \in \{0,1\}$,  let  $\bar{\mathcal{T}}_s(n)$ denote the set of { partial   trees} that satisfy the conditions of $T_s$ trees.

For notational convenience, also set
$$\mathcal{T}(n) \triangleq  \mathcal{T}_0(n) \cup  \mathcal{T}_1(n)
\quad\mbox{and}\quad
\bar{\mathcal{T}}(n) \triangleq  \bar{\mathcal{T}}_0(n) \cup  \bar{\mathcal{T}}_1(n).$$

\item $T \in  \bar{\mathcal{T}}(n) $ is $i$-level if  $\depth(T) \le i+1.$ Set
$$\bar{\mathcal{T}}_s(i:n) \triangleq 
\left\{T_s \in \bar{\mathcal{T}}_s(n)\,:\, T_s \mbox{ is $i$-level}\right\}
\quad\mbox{and}\quad
\bar{\mathcal{T}}(i:n)  \triangleq  \bar{\mathcal{T}}_0(i:n) \cup\bar{\mathcal{T}}_1(i:n).
$$

\item 
Let $T \in \mathcal{T}(n)$. The  {\em $i$-level truncation} of $T,$ denoted by  
$\Trunc^{(i)}(T)$,
%is  the partial binary AIFV  code tree containing all nodes and leaves (i.e. slaves, masters, leaves, and complete internal nodes) at  depth $\leq i+1$ in $T_s$. 
is  the partial   tree that remains after removing all nodes at depth $j > i+1$ from $T.$
\end{itemize}
\end{definition}

\begin{Note}
$\forall T \in {\mathcal{T}}(n),$  $\Trunc^{(i)}(T)\in \bar{\mathcal{T}}(i:n).$ %is an $i$-level  tree.
\end{Note}

\begin{figure*}[t!]
 \center
\includegraphics[width=5.5in]{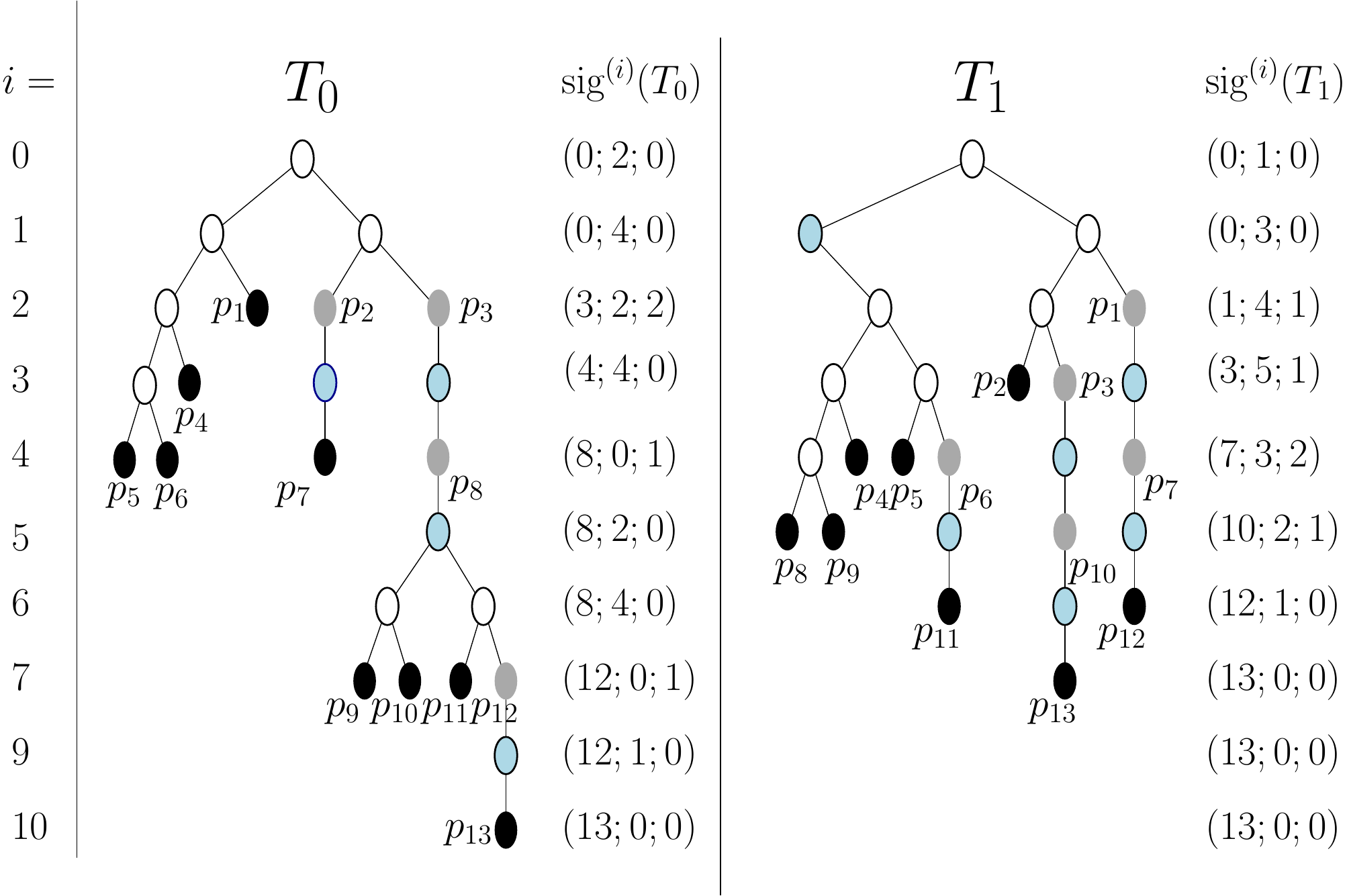}
  \caption{Black nodes are leaves, gray nodes master nodes  and blue ones slave nodes.  Note that on every level, the indices of the leaves are smaller than the indices of the master nodes. Also note that in all cases,
   if $\sig^{(i)}(T_s) = \left(m';p',z'\right)$ and $\sig^{(i+1)}(T_s) = \left(m;p,z\right)$  then
$2m' + 2p' + z'  = 2 m +p,$
$m'+p' \ge m$ and 
$m' \le m-z,$  as required by Lemma \ref{lem:2Deqiv}.
  }
  \label{fig:Sigs}
  \end{figure*}

%$$\begin{array}{|c||c|c|c|c|c|c|c|c|c|c|c|c|c|}
%\hline
%i     & 1  & 2  & 3   & 4   & 5   & 6   & 7   & 8   & 9   & 10  & 11  & 12  & 13  \\ \hline
%p_i  & .2 & .2 & .15 & .1  & .1  & .05 & .05 & .05 & .02 & .02 & .02 & .02 & .02 \\ \hline
%W'_i & .8 & .6 & .45 & .35 & .25 & .2  & .15 & .1  & .08 & .06 & .04 & .02 & 0   \\ \hline
%\end{array}
%$$
\begin{definition}[Signatures and Costs]
See \Cref{fig:Sigs,fig:Expand}.
%Let   $T \in \mathcal{T}(i:n)$.  % of depth $\le i+1$ 

\par\noindent{\bf (a)  $i-$level  Signatures:}
The $i-$level signature of $T$ is the ordered triple
$$\sig^{(i)}(T)\triangleq (m;p;z)$$
\noindent where
\begin{eqnarray*}
m &\triangleq & |\{j\,:\,  d_T(j) \le i\}| = \mbox{\# of codewords on or above level $i$ of $T$},\\
p &\triangleq & \mbox{\# of non-slave nodes on level $i+1$ of $T$},\\
z &\triangleq & |\{j\,:\,  d_T(j) =i \mbox{ and } m_T(j) =1\}| = \mbox{\# of master  nodes on level $i$ of $T$}.
\end{eqnarray*}
%
% $m$ is the number of masters and leaves with depth $j\leq i$ in $T,$ 
%$p$ is the total number of
%non-slave nodes 
%%combined number of master nodes, leaves, and complete internal nodes 
%on level $i+1$ of $T$
%and  $z$ is the number of master nodes   on  level $i$ of $T.$

 Note that
$$ \sig^{(i)}(T)  = \sig^{(i)}\left(\Trunc^{(i)}(T)\right).
$$

\par\noindent{\bf (b)  $i$-level Costs:}

 Let $\sig^{(i)}(T)=(m;p;z).$
The $i$-level costs  of $T$ are
$$\Cost_0^{(i)}(T) \triangleq  i W'_m + \sum_{i=1}^m   d_T(i) p_i+ C \sum_{i=1}^{m-z}  m_T(i) p_i.
$$
and
$$\Cost_1^{(i)}(T) \triangleq  i W'_m + \sum_{i=1}^m  d_T(i)p_i -  C \sum_{i=1}^m  \ell_T(i) p_i.
$$

%Note that
%$$ \Cost^s_i(T_s)  = \Cost^s_i(\Trunc_i(T_s)).
%$$

%\par\noindent{\bf (c) Signatures and Costs:} Let $d = \depth(T).$  Define
%$$\sig(T) = \sig^{(d)}(T)
%$$
%and
%\begin{equation}
%\label{eq:CoC1def}
%\Cost_0(T)  \triangleq  \Cost_0^{(d)}(T) 
%\quad \mbox{and}\quad
%\Cost_1(T)  \triangleq  \Cost_1^{(d)}(T). 
%\end{equation}
\end{definition}
%
%\begin{Note} 
% \Cref{eq:CoC1def}  is a generalization that is consistent with  \Cref{eq:Cost0def,eq:Cost1def}.
%Suppose
%$T_s \in {\mathcal{T}}_s(n),$  
%with $\depth(T_s) = d.$  Then $\sig^{(d)}(T_s)=(n;p; z)$, for some $p,z$ so $W'_n  =0$ and 
%$$
%\Cost_0^{(d)}(T_0) =  d W'_n + \sum_{i=1}^n  d_T(i) p_i+ C \sum_{i=1}^n  m_T(i) = \Cost_0(T_0).
%$$
%Similarly
%$$
%\Cost_1^{(d)}(T_1) =d W'_n +  \sum_{i=1}^n  d_T(i) p_i -  C \sum_{i=1}^n  \ell_T(i) 
%= \Cost_1(T_1).
%$$
%\end{Note}

\begin{figure*}[t]
\vspace*{.1in}

 \center
\includegraphics[width=5.5in]{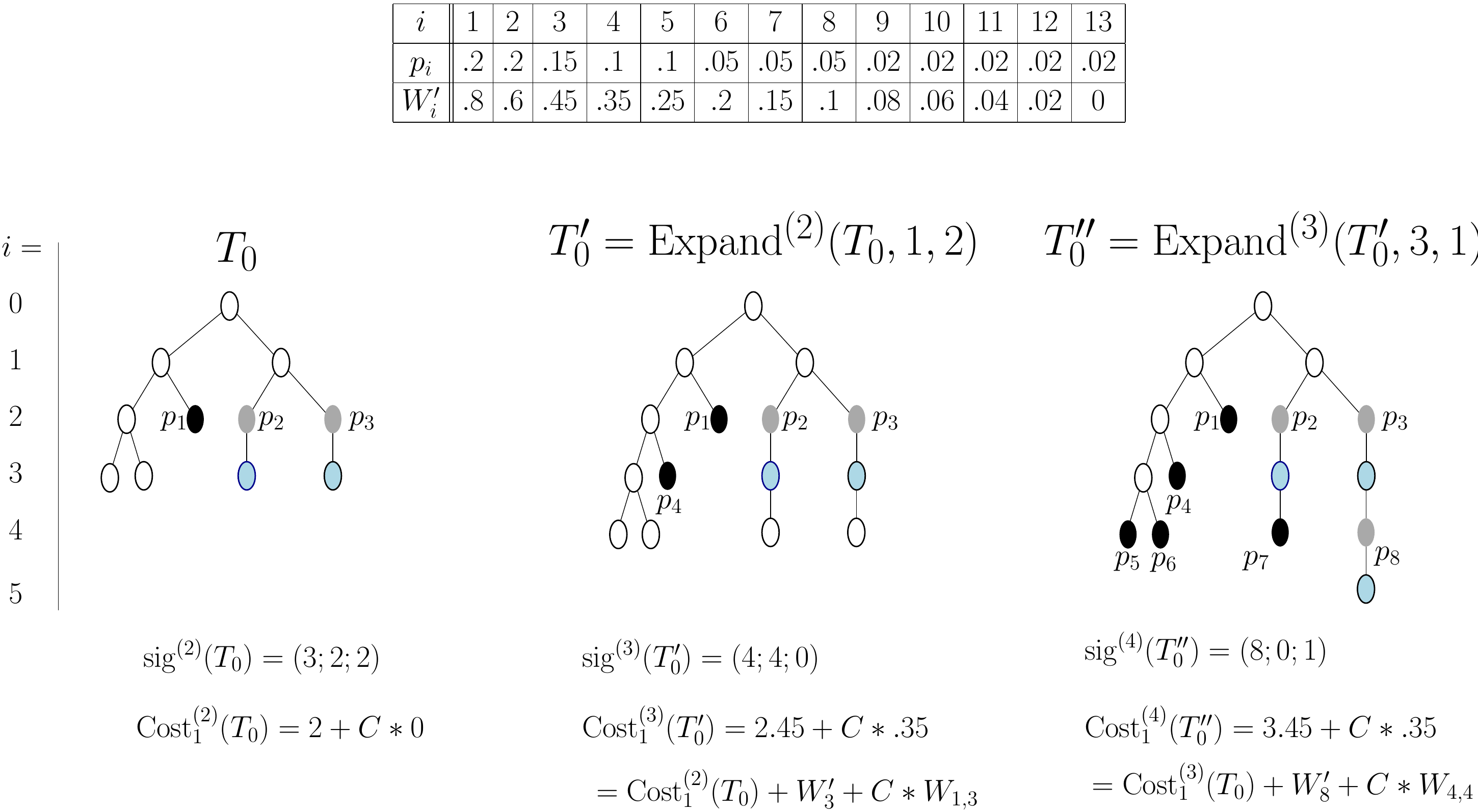}
  \caption{ Illustrations of the $\Trunc$ and $\Expand$ operations and Lemma \ref{lem:expand}. The $p_i$,  $i=1,\ldots,13,$ are given in the table above the trees.
  As examples of the $\Trunc$ operation note that $\Trunc^{(2)}(T''_0) = \Trunc^{(2)}(T'_0) = T_0$ and $\Trunc^{(3)}(T''_0)= T'_0.$
  }
  \label{fig:Expand}
  \end{figure*}

Suppose
$T_s\in \bar {\mathcal{T}}_s(n),$  
with $\depth(T) = d.$  An interesting peculiarity of this definition is that $T$ is an $i$ level tree for all $i \ge d-1$ and it is quite possible that 
$$\Cost_0^{(d-1)}(T_0)  < \Cost_0^{(d)}(T_0) <  \Cost_0^{(d+1)}(T_0) < \cdots
$$
for some indeterminate length chain. The important observation though, is that $\Cost_s^{(i)}(T)$ collapses to $\Cost_s(T)$ for the interesting cases.
\begin{lemma}\ 
\label{lem:levelind}
\begin{itemize}
\item[(a)] Let $T_s\in  {\mathcal{T}}_s(n),$  
with $\depth(T_s) = d.$   \\Then  $\sig^{(d)}(T_s) =(n;0;0)$  and  
$\Cost_s^{(d)}(T_s) = \Cost_s(T_s).$
\item[(b)] Let $T_s\in \bar {\mathcal{T}}_s(n)$ be an $i$-level tree with $\sig^{(i)}(T_s) =(n;0;0)$.\\
%Let $T'_s$ be $T$ with all nodes at depth $i+1$ removed. 
Then $T_s \in  {\mathcal{T}}_s(n)$ with $depth(T_s) =i$.  
%$\sig^{(i)}(T'_s) = (n;0;0)$ and  $\Cost^{(i)}(T'_s) = \Cost^{(i)}(T_s) .$
\end{itemize}
\end{lemma}
\begin{proof}
(a)  By definition, $T_s$ is a $d$-level tree  with no nodes on level $d+1.$

Let $(m,p,z) = \sig^{(d)}(T_s)$. Since $T_s$ contains $n$ codewords, $m =n.$  $T_s$ contains no nodes on level $d+1$, so $p=0.$  Furthermore,  it contains no slave nodes on level $d+1$ so it contains no master nodes on level $d,$ i.e., $z=0.$

Since $W'_n  =0$,
$$
\Cost_0^{(d)}(T_0) =  d W'_n + \sum_{i=1}^n  d_T(i) p_i+ C \sum_{i=1}^n  m_T(i)  p_i= \Cost_0(T_0).
$$
Similarly
$$
\Cost_1^{(d)}(T_1) =d W'_n +  \sum_{i=1}^n  d_T(i) p_i -  C \sum_{i=1}^n  \ell_T(i) p_i
= \Cost_1(T_1).
$$

(b)  $T_s$   contains  no master nodes on level $i$ so  it contains no slave nodes on level $i+1$. It also contains no non-slave nodes on level $i$.  So it contains no nodes on level $i$ and $\depth(T_s) =i.$ $T_s \in  {\mathcal{T}}_s(n)$ by definition.
%
%If $\depth(T_s) =i$ then $T'_s = T_s$ and the lemma is trivially correct.  
%
%Otherwise, $\depth(T_s) = i+1$.  Because $T_s$ contained no master nodes on level $i$, it contains no slave nodes on level $i+1$ so the only nodes on level $i+1$ are children of internal nodes on level $i$.  Removing these nodes from  level $i+1$ transforms their parents on level $i$ into unused leaves.  This does not change the signature, so $\sig^{(i)}(T'_s) = \sig^{(i)}(T_s)$.  By construction,  $\depth(T'_s) =i$ and  $\Cost^{(i)}(T'_s) = \Cost^{(i)}(T_s) $
%and  $T'_s \in  {\mathcal{T}}_s(n)$.
\end{proof}

The definitions and lemmas immediately imply
\begin{corollary}
\label{cor:eqopt}
%$$%T_s(C) = \argmin_{i \ge s; T_s \in \bar{ \mathcal{T}}(i:n)} 
%T_s(C) = \argmin_{\substack { i \ge s\\  T_s \in \bar{ \mathcal{T}}(i:n) \\ \sig^{(i)}(T_s) = (n;0;0)}}
%%{i \ge s; T_s \in \bar{ \mathcal{T}}(i:n)} 
%\Cost^{(i)}_s(T_s)
%$$
\begin{equation}
\label{eq:eqopt}
%T_s(C) = \argmin_{i \ge s; T_s \in \bar{ \mathcal{T}}(i:n)} 
T_s(C) = \argmin_{\substack { T_s \in \bar{ \mathcal{T}}(n) \\ \exists i \mbox{ s.t. } T_s \in \bar{ \mathcal{T}}(i:n)  \mbox{ and } \sig^{(i)}(T_s) = (n;0;0)}}
\Cost^{(i)}_s(T_s)
\end{equation}
\end{corollary}

The next definition introduces the initial conditions for the dynamic programs.
\begin{definition} 
\label{def:firstlev}
See \Cref{fig:Initials}.
Set
\begin{eqnarray*}
%I^0 &=& \{(0;2;0)\},\\
%I_0 &=& \{(0;4;0),\,  (1;2;0),\,   (1;2;1),\, (2;0,1) ,\, (2;0,2)\},\\
I_0 &=& \{(0;2;0),\, (1;0,1)\},\\
I_1 &=&\{ (0;3;0),\, (1;1;0),\,  (1;1;1)\}.
\end{eqnarray*}
 Note that if $(m;p;z) \in I_0$,  there exists a unique $0$-level tree $T_s \in  \bar {\mathcal{T}}_0(n)$ satisfying
$\sig^{(0)}(T_0)  = (m;p;z).$ 

Similarly,  if $(m;p;z) \in I_1$,  there exists a unique $1$-level tree $T_s \in  \bar {\mathcal{T}}_1(n)$ satisfying
$\sig^{(1)}(T_1)  = (m;p;z).$ 

Let 
 $T_s(m;p;z)$ denote this unique tree and $\bar c_s(m;p;z) = \Cost_s^{(s)}(T_s(m;p;z)).$
\end{definition}

\begin{figure*}[t]
 \center
\hspace*{-.4in}\includegraphics[width=5.5in]{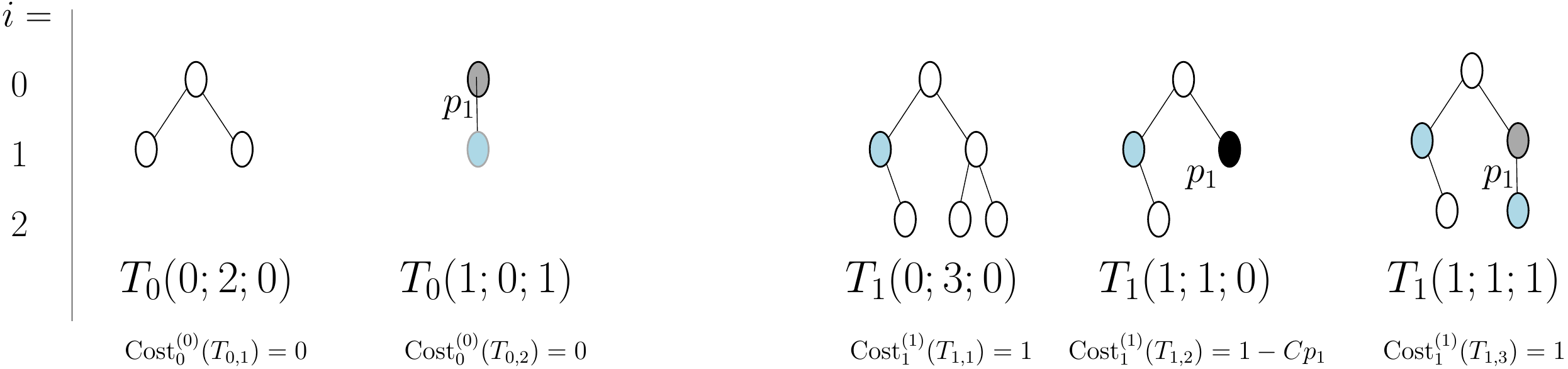}
  \caption{The initial trees introduced in Definition \ref{def:firstlev}.   % Blue, grey and Bue nodes are slave nodes while grey and black nodes are, respecively, master and leaf nodes (with associated codewords).
  Note that the definition of $T_0$ trees permit the root to be a master node or an internal node, while the definition of $T_1$ trees requires that the root be an internal node.
  %Note that all 5 $T_0$ trees satisfy $\Cost_1(T_0(m;p;z))=1.$
  }  \label{fig:Initials}
   \end{figure*}

The following lemma is true by observation
\begin{lemma}
\label{lem:firstlev}  Let $n > 2$.\\
If $T_0 \in  \bar {\mathcal{T}}_0(n)$ with  $\depth(T_0) \ge 0$, then
$\sig^{(0)}(T_0)  \in I_0.
$\\
If $T_1 \in  \bar {\mathcal{T}}_1(n)$ with  $\depth(T_1) \ge 1$, then
$\sig^{(1)}(T_1)  \in I_1.
$
\end{lemma}

 \begin{Note}
  The reason for starting with  $\sig^{(1)}(T_1)$ instead of  $\sig^{(0)}(T_1)$ is because the 
 root of  a $T_1$ tree  is ``unusual'', being  a complete node with a slave child, the only time this combination can occur.
  % This is the only case of a slave child having a non-master parent.  
    By definition, $\sig^{(0)}(T_1)=(0;1;0)$. This is misleading because it loses the information about the unusual slave node on level $1$.  We therefore only start looking  at signatures of $T_1$ trees from level $1$.  
   %For  notational consistency,  we  use the same convention for $T_0$ trees as well.
\end{Note}

\begin{definition} See \Cref{fig:Expand}.
Let $T' \in \bar{\mathcal{T}}(i:n)$ %be an $i$-level tree 
 satisfy $\sig^{(i)}(T')=(m';p';z')$
and
 \begin{equation}
 \label{eq:pdefexXX}
 e_0,\,e_1 \ge 0 \quad \mbox{ \rm such that }\quad  e_0 + e_1 \le p'.
 \end{equation} 
%Let  $0\leq e_0,e_1$  and   $e_0 + e_1 \le p'.$ 
 Define  the {\em $(e_0,e_1)-$expansion} of $T'$ 
as the unique  tree  $$T=\Expand^{(i)}(T', e_0, e_1)$$
in which
\begin{itemize}
\item the  first $i$-levels of $T$ are identical to those of $T'.$
\item  $e_0$ of the $p'$ non-slave nodes on level $i+1$ of $T'$ are set as leaves associated with $a_{m'+1}.\ldots,a_{m'+e_0}.$
\item  $e_1$  non-slave nodes on level $i+1$ of $T'$ are set as master nodes associated with $a_{m'+e_0+1}.\ldots,a_{m'+e_0+e_1}$
(with corresponding slave nodes created on level $i+2$).
\item the remaining $p'-e_0-e_1$ non-slave nodes on level $i+1$ of $T'$ become complete internal nodes, creating $2(p'-e_0-e_1)$ non-slave  nodes on level $i+2$. These are in addition to the $z'$ non-slave children on level $i+2$ of the $z'$ slave nodes on level $i+1.$ 
\end{itemize}
Note that this definition implies that $\sig^{(i+1)}(T) = (m;p;z)$ where
 \begin{eqnarray}
 m  &=& m'+e_0+e_1, \label{eq:pdefmxX}\\
 z   &=& e_1,\label{eq:pdefzX}\\
 p &=&  z' + 2(p'-e_0-e_1). \label{eq:pdefpX}
 \end{eqnarray}
\end{definition}

\begin{lemma}
\label{lem:TruncExpand} \ 
%Let $T' \in \bar{\mathcal{T}}(i:n)$  %be an $i$-level tree 
%and
%$T \in \bar{\mathcal{T}}(n).$
% and, for $i >0,$   $(m^{(i)};p^{(i)};z^{(i)}) = \sig^{(i)}(\Trunc^{(i)}(T)).$ Then
\begin{itemize}
\item[(a)] Let  $T' \in \bar{\mathcal{T}}(i:n)$.  If $\sig^{(i)}(T') = (m';p';z')$ and  $(e_0,e_1)$ satisfies \Cref{eq:pdefexXX}, then %  Equations (\ref{eq:pdefe})-(\ref{eq:pdefp}),  then
$$T'= \Expand^{(i)} 
\left(
T'',\, e_0, \, e_1
\right) \in \bar{\mathcal{T}}(i+1:n).
$$

\item[(b)]Let  $T \in \bar{\mathcal{T}}(n).$  For $i \ge 0,$   set $\left(m^{(i)};p^{(i)};z^{(i)}\right) = \sig^{(i)}\left(\Trunc^{(i)}(T)\right).$ Then
$$\Trunc^{(i+1)}(T) = \Expand^{(i)} 
\Bigl(
\Trunc_{i}(T),\, e_0, \, e_1
\Bigr)
$$
where $e_0 =m^{(i+1)} - m^{(i)} -z^{(i+1)}$ and $e_1 = z^{(i+1)}.$
\end{itemize}
\end{lemma}
\begin{proof}
(a) follows from the fact that Definition \ref{eq:pdefe} maintains the validity of properties (P1) and (P2) of Lemma \ref{lem:DPStructure} and that $\depth(T'') \le i+2.$
(b) just follows directly from the definitions.
\end{proof}

Part (b) implies that any  tree $T \in \bar{\mathcal{T}}(n)$ can be grown level by level via expansion operations.

%*************  Refer to Signature Set Definition *****************
Now recall from
\Cref{def:Exp1}
the definition of the signature set $\CALS_n$ and the operation $\rightarrow$.
\begin{lemma}
\label{lem:expand}
Let 
$T' \in \bar{\mathcal{T}}(i:n)$  %be an $i$-level tree 
with  $\sig^{(i)}(T') = \alpha' = (m';p';z')$.% and
%$T$ be an $(i+1)$-level tree with  $\sig_{i+1}(T) = \alpha' = (m;p;z)$. Then

\begin{itemize}

\item[(a)] 
Let $(e_0,e_1)$ satisfy \Cref{eq:pdefe}.\\ % $e_0,e_1$ satisfy
%Equations (\ref{eq:pdefe})-(\ref{eq:pdefp}),
 %\Cref{eq:pdefe}, %for some $0 \le e_0,e_1$ and $e_0 + e_1 \le p$, 
Let $T=\Expand^{(i)}(T', e_0, e_1)$ and 
$\alpha=(m;p;z) = \sig^{(i+1)}(T).$

Then
$ \alpha' \rightarrow \alpha$.

\item[(b)] Let $\alpha=(m;p;z)$.    If $ \alpha' \rightarrow \alpha$,  let 
$e_0,e_1$ be the unique values satisfying  Equations (\ref{eq:pdefe})-(\ref{eq:pdefp}) and set $T=\Expand^{(i)}(T', e_0, e_1)$.

Then $\alpha=\sig^{(i+1)}(T)$.

%\item[(a)] $T=\Expand_i(T', e_0, e_1)$ for some 
% some $0 \le e_0,e_1$ and $e_0 + e_1 \le p$ if and only if 
%$ \alpha' \rightarrow \alpha$.
%
%Furthermore,  if $ \alpha' \rightarrow \alpha$ there exists a unique pair $e_0,e_1$ such that $T=\Expand_i(T', e_0, e_1)$.
%

%Then $(m';p';z')$, $(m;p;z)$ satisfy Equations (\ref{eq:pdefm})-(\ref{eq:pdefp})

%  Then 
%$$\sig_{i+1}(T)=(m'+e_0+e_1; e_{1}; z'+2(p'-e_0-e_1)).$$
\item[(c)]   If $T=\Expand^{(i)}(T', e_0, e_1)$ with   $\alpha=(m;p;z) = \sig^{(i+1)}(T)$, then
  %if they are type-$0$ trees then 
$$
\Cost_0^{(i+1)}(T) = \Cost_0^{(i)}(T') + c_0(\alpha',\,  \alpha) % W'_{m'} + C W_{m'-z',m'},
$$
%while if  they are type-$1$ trees then 
$$
\Cost_1^{(i+1)}(T) = \Cost_1^{(i)}(T') + c_1(\alpha',\,  \alpha) %W'_{m'} - C W_{m',m-z}.
$$
\end{itemize}

\end{lemma}
\begin{proof} \ 

(a) %  First assume $T=\Expand_i(T', e_0, e_1)$ where $\sig_i(T) = \alpha.$ 
This follows directly from the definition of  $T=\Expand^{(i)}(T', e_0, e_1)$.

%The definition of $\Expand$ immediately implies  $m= m' + e_0 + e_1$ and  $z=e_1.$
%
%The non-slave nodes on the $(i+2)nd$ level of $T$ are either children of the $z'$ slave nodes on level $i+1$ or children of the 
%$p' - e_0 - e_1$ non-leaf, non-master nodes on level $i+1$.  Thus 
%$p= z'+2(p'-e_0-e_1))$,  Equations (\ref{eq:pdefm})-(\ref{eq:pdefp}) are satisfied and thus $ \alpha' \rightarrow \alpha$.

\medskip

(b) From the definition of  $ \alpha' \rightarrow \alpha$ there exist appropriate $e_0,e_1$ satisfying Equations (\ref{eq:pdefe})-(\ref{eq:pdefp}).
Then  $ T=\Expand^{(i)}(T', e_0, e_1)$ has  $\sig^{(i+1)}( T) = (m;p;z).$ 
%The uniqueness of $e_0,e_1$ follows from the fact that 
%Equations (\ref{eq:pdefm})-(\ref{eq:pdefz}) have a unique solution.

\medskip

(c) 
From the definitions of signatures and expansions
$$
\sum_{j=1}^{m'}  d_{T}(j)p_j = \sum_{j=1}^{m'}  d_{T'}(j)p_j
\quad\mbox{and}\quad
\sum_{j=m'+1}^{m}  d_{T}(j)p_j = (i+1) \sum_{j=m'+1}^{m} p_j  = (i+1) W_{m',m}.
$$
Furthermore, from Lemma \ref{lem:DPStructure}  (P1), (P2),
the master  nodes  on level $i$ correspond to  $a_{m'-z'+1},\ldots, a_{m'}$.  Thus (again also using the definition of expansion)
$$ 
\sum_{j=1}^{m'-z'}  m_{T}(j)p_j = \sum_{j=1}^{m'-z'}  m_{T'}(j)p_j
\quad\mbox{and}\quad  \sum_{j=m'-z'+1}^{m-z}  m_T(j) p_j  =\sum_{j=m'-z'+1}^{m-z}  p_j  = W_{m'-z',m'}.$$
Then
\begin{eqnarray*}
\Cost_0^{(i+1)}(T) &=&  (i+1) W'_{m} + \sum_{j=1}^{m}  d_{T}(j)p_j+   C\sum_{j=1}^{m-z}   m_T(j)p_j\\
			 &=&  (i+1) W'_{m}  + (i+1)  W_{m',m}  + \sum_{j=1}^{m'}  d_{T'}(j)p_j \\
			 && \hspace*{.8in} +   C  %\hspace*{-.05in}
			 \sum_{j=1}^{m'-z'}   m_{T'}(j)p_j
			   +  C  %\hspace*{-.1in}
			     \sum_{j=m'-z'+1}^{m-z} m_T(j) p_j \\
			 &=&  (i+1) W'_{m'} + \sum_{j=1}^{m'}  d_{T'}(j)p_j+   C \hspace*{-.05in}\sum_{j=1}^{m'-z'}   m_{T'}(j)p_j
			 + C  \hspace*{-.1in}\sum_{j=m'-z'+1}^{m-z} m_T(j) p_j \\
			 &=& i W'_{m'}  +   \sum_{j=1}^{m'}  d_{T'}(j)p_j  +  C \sum_{j=1}^{m'-z'}   m_{T'}(j)p_j
			  + W'_{m'} +  C W_{m'-z',m'}\\
%			 &=& \Cost^0_i(T') + W'_{m'} + C W_{m'-z',m'}\\
			 &=& \Cost_0^{(i)}(T')  +  c_0(\alpha',\,  \alpha).
\end{eqnarray*}
%Now  suppose that $T',T$  are type-$1$ trees. 

From Lemma \ref{lem:DPStructure}   (P1), (P2),
the leaves on level $i+1$ of $T$ correspond to
$a_{m'+1},\ldots, a_{m-z}$.
% and the master nodes on level $i+1$ of $T'$ correspond to
%$a_{m-z+1},\ldots, a_{m}$. Similarly Furthermore,  nodes  $a_{m'-z'+1},\ldots, a_{m'}$  are the master nodes on level $i.$ 
%Again, from the earlier  discussion,
Thus (again also using the definition of expansion)
$$  
\sum_{j=1}^{m'}  \ell_{T}(j)p_j = \sum_{j=1}^{m'}  \ell_{T'}(j)p_j
\quad\mbox{and}\quad
 \sum_{j=m'+1}^{m'} \ell_T(j)p_j  =  \sum_{j=m'+1}^{m-z} p_j  = W_{m',m-z}.$$
Then
\begin{eqnarray*}
\Cost_1^{(i+1)}(T) &=&  (i+1) W'_{m} + \sum_{j=1}^{m}  d_{T}(j)p_j  -  C \sum_{j=1}^{m}   \ell_T(j)p_j\\
			 &=&  (i+1) W'_{m} + (i+1)  W_{m',m}  + \sum_{j=1}^{m'}  d_{T'}(j)p_j-  C \sum_{j=1}^{m'}   \ell_{T'}(j)p_j
			  -   C \sum_{j=m'+1}^{m}  \ell_T(j)p_j\\
			 &=& i W'_{m'}  +   \sum_{j=1}^{m'}  d_{T'}(j)p_j -  C \sum_{j=1}^{m'}   \ell_{T'}(j)p_j
			 + W'_{m'} -  C W_{m',m-z}\\
%			 &=& \Cost^1_i(T') + W'_{m'} - C W_{m',m-z}\\
			 &=& \Cost_1^{(i)}(T')  +  c_1(\alpha',\,  \alpha).
\end{eqnarray*}
\end{proof}

Combining \Cref{lem:firstlev,lem:TruncExpand,lem:expand}  immediately imply a direct relationship between paths in the Signature Graph and building  a tree level-by-level.
\begin{corollary}  Fix $s \in \{0,1\}.$
\label{cor:interp}

\par\noindent{(a)}
Let $T  \in \bar{\mathcal{T}}_s(i:n)$  %be an $i$-level tree 
and, for all $s \le j \le i$ set   
$$T^{(j)} =\Trunc^{(j)}(T)
\quad\mbox{and}\quad
%\alpha^{(j)} = \left(m^{(j)};p^{(j)};z^{(j)}\right) = \sig^{(j)}\left(T^{(j)}\right).
\alpha^{(j)} =  \sig^{(j)}\left(T^{(j)}\right).
$$
Then
\begin{itemize}
\item   $\alpha^{(s)} \in I_s;$\quad  $T^{(s)} =  T_s\left(\alpha^{(s)} \right);$\quad $ \Cost_s\left(T^{(s)}\right) =c_s \left(\alpha^{(s)}\right);$
\item $\forall s \le j  < i$,$\quad  \alpha^{(j)}  \rightarrow  \alpha^{(j+1)}$
\item  $ \Cost_s^{(i)}\left(T^{(i)}\right) =\bar c_s \left(\alpha^{(s)}\right) + \sum_{j=s}^{i-1} c_s\left(\alpha^{(j)},\,\alpha^{(j+1)}\right)$
%\item $ \Cost_s\left(T^{(s)})$  %= c_s\left\alpha^{(s)}( \right);$ 
%   $ \Cost_s^{(i)}(T) =  \Cost_s^{(s)}(T)$
\end{itemize}

\par\noindent{(b)} Let $\left\{ \alpha^{[j]}\right\}_{j=s}^i \subset \CALS_n$ such that  $\alpha^{[s]} \in I_s$ and   for all $s \le j  < i$, $\alpha^{[j]}  \rightarrow  \alpha^{[j+1]}$.
Then there exists an $i$ level tree $T  \in \bar{\mathcal{T}}_s(n)$ such that, using the definitions from part (a),  $\alpha^{(j)} = \alpha^{[j]}.$

\end{corollary}

{\em \small Note: the condition $ s \le j$ reflects the fact that, from \Cref{def:firstlev},  \Cref{lem:firstlev}  and the explanatory note following
 \Cref{lem:firstlev}, 
the initial condition for $T_0$ requires $j \ge 0$ and the initial condition for $T_1$ requires $j \ge 1.$
}

\medskip

This Corollary motivates the original definition of the  $\OPT_s(\alpha)$ tables.

\begin{lemma} 
Fix $s \in \{0,1\}$ and define initial signatures  $I_s$  with associated $\bar c_s(\alpha)$ for $\alpha \in I_s$ as in \Cref{def:firstlev}.
Let $\OPT_s(\alpha)$   and $\Pred_s(\alpha)$ 
be as introduced in \Cref{def:OPT Tables}.

Then, for all $\alpha \in \CALS_n,$
\begin{equation}
\label{eq:mainDP1}
\OPT_s(\alpha)  =
%\min\left\{\Cost^{(i)}(T_s) \,:\,  \mbox{\rm for some   } i \ge s  \ {\rm  and  } T_s  \in \bar{\mathcal{T}}_s(i:n)  \  \mbox{\rm   such that } \sig( T_s) = \alpha  \right\}.
\min\left(
\bigcup_{i \ge s} \left\{\Cost_s^{(i)}(T_s) \,:\, T_s  \in {\mathcal{T}}_s(i:n) \mbox{ \rm and }   \sig^{(i)}( T_s) = \alpha \right\}
\right).
\end{equation}
Furthermore,  an $i \ge s$ and $T_s \in {\mathcal{T}}_s(i:n)$ satisfying
\begin{equation}
\label{eq:mainDP2}
\sig^{(i)}(T_s) = \alpha
\quad\mbox{and}\quad 
\Cost_s^{(i)}(T_s)  =\OPT_s(\alpha)
\end{equation}
can be constructed in $O(i)$ time using the $\Pred_s(\,)$ entries.
%{eq:eqopt}
\end{lemma}
\begin{proof}
Recall the interpretation of   $OPT_s(\alpha)$ given after \Cref{def:OPT Tables}.  Consider the $\alpha$ as nodes in a directed graph 
with edge costs defined by $c_s(\alpha',\,  \alpha)$ except  that edges from $(0;0;0)$ to $\alpha \in I_s$  have cost 
$\bar c_s(\alpha)$ and all other undefined edge costs are set to $\infty.$  Then $ \OPT_s(\alpha)$ is just the cost of the shortest path from $(0;0;0)$ to $\alpha$.

\Cref{cor:interp}(a)  then implies that if $T_s  \in \bar{\mathcal{T}}_s(i:n)$ with $\sig^{(i)}(T_s) = \alpha,$ then there exists a path from  $(0;0;0)$ to $\alpha$ with cost $\Cost_s^{(i)}(T_s).$

In the other direction, \Cref{cor:interp}(b) implies that if $P$ is a $i$-edge path from $(0;0;0)$ to $\alpha$, then there exists $T_s  \in \bar{\mathcal{T}}_s(i:n)$ with $\sig^{(i)}(T_s) = \alpha,$  and
$\Cost_s^{(i)}(T_s)$ equal to the cost of the path.

This proves \Cref{eq:mainDP1}.

The actual tree $T_s$  satisfying  \Cref{eq:mainDP1} can be found by following the  $\Pred_s(\,)$ values backwards from $\alpha$ until reaching $\alpha' \in I_s$.  This provides a path from $(0;0;0)$ to $\alpha$ with cost $\OPT_s(\alpha)$. This path can be translated into $T_s$  via \Cref{cor:interp}(b). 
\end{proof}
\Cref {cor:eqopt}  then immediately implies
\begin{corollary}
Fix $s \in \{0,1\}$.   Then
$$\min_{T_s \in \mathcal{T}_s(n)} \left\{\Cost_s(T_s:C) \right\} = \OPT_s(n;0;0).$$
Furthermore,  if $i \ge s$ and $T_s \in {\mathcal{T}}_s(i:n)$ are such that $Cost_s^{(i)}(T_s)  =\OPT_s(n;0;0)$, then $T_s(C) = T_s.$
\end{corollary}

In words, the Corollary states that $T_s(C)$ can be found by filling in the $\OPT_s(\,)$ table and then using the $\Pred_s(\,)$ entries to construct the tree corresponding to $\OPT_s(n;0;0).$ Since \Cref{sec:The DP} gives an $O(n^3)$ algorithm for filling in the $\OPT_s(\,)$ and $\Pred_s(\,)$ tables, this leads to the desired $O(n^3)$ algorithm for solving the original problem.